\documentclass[manuscript,screen]{acmart}

\AtBeginDocument{%
  }

\setcopyright{acmlicensed}
\copyrightyear{2025}
\acmYear{2025}
\acmDOI{XXXXXXX.XXXXXXX}
\acmConference[Conference acronym 'XX]{Make sure to enter the correct
  conference title from your rights confirmation emai}{June 03--05,
  2018}{Woodstock, NY}
\acmISBN{978-1-4503-XXXX-X/18/06}

\usepackage{bm}
\usepackage{algorithm}
\usepackage[noend]{algorithmic}
\usepackage{amsmath}
\usepackage{amsfonts}
\usepackage{bbm}
\usepackage{subcaption}
\usepackage{multirow}

\newcommand{\A}{\bm{A}}
\newcommand{\B}{\bm{B}}

\newcommand{\plmi}{\,$\pm$\,}

\renewcommand{\algorithmicrequire}{\textbf{Input:}}
\renewcommand{\algorithmicensure}{\textbf{Output:}}

\begin{document}

\title{Balancing Fairness and High Match Rates in Reciprocal Recommender Systems: A Nash Social Welfare Approach}

\author{Yoji Tomita}
\email{tomita_yoji@cyberagent.co.jp}
\orcid{0009-0007-6828-4424}
\affiliation{%
  \institution{CyberAgent, Inc.}
  \city{Tokyo}
  \country{Japan}
}
\author{Tomohiko Yokoyama}
\email{tomohiko_yokoyama@mist.i.u-tokyo.ac.jp}
\orcid{0009-0004-7641-029X}
\affiliation{%
  \institution{School of Information Science and Technology, The University of Tokyo}
  \city{Tokyo}
  \country{Japan}
}

\renewcommand{\shortauthors}{Y. Tomita and T. Yokoyama}

\begin{abstract}
Matching platforms, such as online dating services and job recommendations, have become increasingly prevalent. For the success of these platforms, it is crucial to design reciprocal recommender systems (RRSs) that not only increase the total number of matches but also avoid creating unfairness among users.

In this paper, we investigate the fairness of RRSs on matching platforms. From the perspective of fair division, we define the users' opportunities to be recommended and establish the fairness concept of envy-freeness in the allocation of these opportunities. We first introduce the Social Welfare (SW) method, which approximately maximizes the number of matches, and show that it leads to significant unfairness in recommendation opportunities, illustrating the trade-off between fairness and match rates. To address this challenge, we propose the Nash Social Welfare (NSW) method, which alternately optimizes two NSW functions and achieves nearly envy-free recommendations. We further generalize the SW and NSW method to the $\alpha$-SW method, which balances the trade-off between fairness and high match rates. Additionally, we develop a computationally efficient approximation algorithm for the SW/NSW/$\alpha$-SW methods based on the Sinkhorn algorithm. Through extensive experiments on both synthetic datasets and two real-world datasets, we demonstrate the practical effectiveness of our approach.
\end{abstract}



\maketitle

\section{Introduction}\label{section:Introduction}
In recent decades, matching platforms such as online dating services and job recommendation systems have proliferated rapidly across many countries. 
In contrast to conventional recommender systems that suggest items (e.g., movies, music) to users, these platforms employ reciprocal recommender systems (RRSs) where users on both sides receive recommendations and can express mutual interest, creating a fundamentally different recommendation paradigm~\cite{pizzato2010recon,mine2013reciprocal,palomares2021reciprocal}.
Critically, this reciprocal nature means that designing recommendations becomes equivalent to allocating opportunities to be recommended among users on both sides.

A primary objective for matching platforms is maximizing successful matches between users.
To achieve this, it is essential to avoid excessive concentration of recommendation opportunities (i.e., the chances of being recommended to others) on popular users.
This limitation arises since even if a highly popular user is recommended to many users on the other side and receives numerous expressions of interest, the number of connections they can meaningfully engage with is physically limited.
Several recent studies show that practical online dating recommendation algorithms suffer from this popularity bias~\cite{chen2023reducing,celdir2024popularity}. 
To address this issue, some reciprocal recommendation algorithms have been developed to avoid popularity bias and increase the total number of matches~\cite{chen2023reducing,tomita2022matching,Tomita2023,su2022optimizing}. 

In parallel with maximizing matches, ensuring fairness among users represents a crucial consideration for matching platforms.
This is particularly important in RRSs where both users receiving recommendations and those being recommended are individuals who can perceive and react to unfair treatment.
Unfairness among users on both sides can directly affect user satisfaction and the platform's reputation.
Furthermore, fair recommendation distribution may inherently reduce popularity bias, thereby also contributing to the efficiency objective of maximizing matches in real-world applications.

While fairness has emerged as a critical research area in conventional recommender systems~\cite{Li2022Fairness,Wang2023FairnessRecommender,DeldjooJannach2024}---encompassing user-side fairness~\cite{ekstrand2018all,wang2021user}, item-side fairness~\cite{beutel2019fairness,zhu2020measuring,SaitoJoachimsKDD2022}, both-side fairness~\cite{WuCaoXuTan2021},
group fairness~\cite{kearns2018preventing,kearns2019empirical}, and individual fairness~\cite{biega2018equity,singh2018fairness}---studies addressing fairness in RRSs remain limited~\cite{xia2019we,Virginie2021}.
Additionally, the dual role of users on online matching platforms---simultaneously acting as active users receiving and acting on recommendations (``likes'' or applications) and passive users being recommended to others---creates a fundamentally different structure from conventional recommender systems.
This unique characteristic necessitates new definitions and analyses of fairness in RRSs.

To address this challenge, we observe that the RRS design problem can be viewed as allocating limited recommendation opportunities among competing users. This perspective naturally leads us to consider fairness frameworks from resource allocation in economics. Specifically, we adopt the concept of \emph{envy-freeness} from fair division within economics and algorithmic game theory~\cite{Foley1967,Varian1974,Budish2011,SaitoJoachimsKDD2022}\footnote{In fair division, which studies how to allocate divisible or indivisible resources fairly among agents, one agent is said to have \emph{envy} toward another if they find another agent's bundle of resources more desirable than their own bundle, and an allocation is said to be \emph{envy-free} if no agent envies another.}.

We aim to ensure that users do not envy others' recommendation opportunities. For example, consider two users with similar characteristics or popularity levels---if one user receives significantly better recommendations (e.g., appearing frequently at the top of the other side's recommendation lists) and consequently gets more matches, this creates an unfair situation. The disadvantaged user would naturally feel envy toward the favored user's treatment. 
Motivated by these observations, we introduce a novel fairness notion in RRSs based on envy-freeness. Building on \citet{SaitoJoachimsKDD2022}, who defined item-side envy-freeness in conventional recommender systems, we extend this concept to both-side fairness in RRSs.

We further propose a natural formulation for RRSs on online matching platforms. Specifically, we consider a model where users' preference probabilities are given through data-driven methods, and when we provide recommendations, the resulting expected matches can be computed.
Our key research questions are: How can we achieve fair recommendations in this setting? Can we develop fairness-aware algorithms that are competitive with baseline approaches in RRS research? 
Furthermore, can we achieve fairness without significantly sacrificing the number of matches?

\subsection{Contributions}
Our contributions are as follows:

First, we formalize the reciprocal recommender system (RRS) problem for online matching platforms and introduce fairness concepts for recommendation opportunities based on envy-freeness from fair division.
While \citet{SaitoJoachimsKDD2022} defined envy-freeness for the item-side in conventional recommender systems, we extend this concept to both side fairness in RRSs, where users on both sides can experience envy regarding their recommendation opportunities.
We further introduce Pareto optimality as an efficiency concept and the Gini index as an additional fairness metric from economics to assess RRSs. These serve as analytical tools in our experimental evaluation.

Second, we introduce a method to maximize expected matches in our proposed model and prove that match-maximizing recommendations satisfy Pareto efficiency.
Importantly, since this involves maximizing matches between users on both sides, the optimization problem is non-convex. We propose an alternating heuristic optimization approach, which we call the social welfare (SW) method, and experimentally demonstrate its effectiveness in achieving high match rates.
However, through theoretical analysis and numerical experiments, we demonstrate that this approach leads to significant unfairness among users, thereby revealing a trade-off between the total number of matches and fairness.

Third, we propose a Nash social welfare (NSW) method that achieves approximately envy-free recommendations while maintaining competitive match rates as shown in Theorem~\ref{thm:envy-free-policy}.
Building on \cite{SaitoJoachimsKDD2022}, whose approach relied on convex optimization, we develop new techniques to handle the non-convex optimization challenges in RRSs by alternately maximizing two NSW functions using the Frank-Wolfe algorithm. 
Remarkably, our experimental evaluation demonstrates that this method ensures almost zero envy while providing competitive match performance.

Fourth, we propose an optimization framework that enables a parameterized transition between SW and NSW. While SW maximizes matches at the expense of fairness and NSW ensures fairness but may result in slightly fewer matches, our framework introduces a tunable parameter $\alpha$ to balance these objectives, allowing for flexible trade-offs between fairness and efficiency.
We call this the $\alpha$-SW method.

Fifth, we develop a fast algorithm based on the Sinkhorn algorithm~\cite{Sinkhorn1967,Cuturi2013} for both SW and NSW methods. 
Since our recommendation model represents each user's ranking positions in others' recommendation lists using doubly stochastic matrices, the Sinkhorn algorithm---which efficiently optimizes over doubly stochastic constraints---is naturally applicable to our setting.
By incorporating the Sinkhorn algorithm for the NSW method, we enable scalable computation that maintains both fairness guarantees and competitive match performance, even for substantially larger problem instances. Furthermore, the algorithm can be parallelized using GPU acceleration.

Finally, we conduct comprehensive experiments using both synthetic datasets and real-world data from a Japanese online dating platform.
We compare our methods against baseline approaches derived from existing research. 
Our experiments demonstrate that the NSW method achieves statistically significant improvements in fairness while maintaining competitive match performance compared to the SW method. All code is publicly available online\footnote{Code repository: \url{https://github.com/CyberAgentAILab/FairReciprocalRecommendation}}.

\subsection{Paper Structure and Differences from the Conference Version}

This article is an extended version of the conference article published in RecSys'24~\cite{TomitaYokoyama2024}. 
The conference paper includes the following components: the problem formalization for RRSs and fairness concepts based on envy-freeness, presented in Sections~\ref{sec:model},~\ref{sec:social_welfare}, and~\ref{sec:fairness_axioms}. The SW and NSW methods using linear programming (LP) solvers are described in Sections~\ref{sec:SW} and~\ref{sec:NSW}, respectively.
The experimental evaluation includes baseline methods in Section~\ref{sec:experiments:Baseline_Methods}, synthetic data experiments in Section~\ref{sec:experiments:Synthetic_Data_Experiment_I}, and real-world experiments using data from a Japanese online dating platform in Section~\ref{sec:experiments:Real-World_Data_Experiment}.

In this extended journal paper, we make several additional contributions beyond the conference version. 
We introduce Pareto optimality and the Gini index as additional evaluation metrics from economics in Sections~\ref{sec:Pareto} and~\ref{sec:experiments:Gini_index}, respectively. 
Furthermore, in Section~\ref{sec:trade-off-fairness-socialwelfare}, we demonstrate the $\alpha$-SW method that enables a parameterized transition between SW and NSW. 
Since the original conference paper used an LP solver for alternating maximization, large-scale computation was challenging. This extended paper addresses this limitation with a Sinkhorn-based algorithm presented in Section~\ref{sec:Sinkhorn}. 
The experimental evaluation is expanded with new studies: $\alpha$-SW method experiments in Section~\ref{sec:experiments:Synthetic_Data_Experiment_II}, Sinkhorn algorithm validation in Section~\ref{sec:experiments:Sinkhorn}, and an additional real-world dataset from a speed dating experiment in Section~\ref{sec:experiments:Real-World_Data_Experiment}.
Moreover, we include additional experiments in the Appendix: robustness to preference estimation errors in Appendix~\ref{appendix:robustness} and sensitivity to hyperparameters in Appendix~\ref{appendix:sensitivity}.

\section{Further Related Work}

\paragraph{Reciprocal Recommender Systems (RRSs)}
Reciprocal recommender systems (RRSs) have attracted increasing attention, with comprehensive surveys provided by~\citet{palomares2021reciprocal}.
Early research primarily focused on preference estimation between users on opposite sides~\cite{pizzato2010recon,xia2015reciprocal,neve2019aggregation,neve2019latent}.
Some recent works have studied the post-processing phase to compute reciprocal scores of user pairs.
\citet{su2022optimizing} formalize the problem of the post-processing phase in reciprocal recommender systems and propose a new algorithm based on convex programming.
\citet{Tomita2023} study the similar model of \citet{su2022optimizing} and consider the TU matching algorithm from the matching theory in economics, whose application in online dating is also studied by \citet{tomita2022matching} and \citet{chen2023reducing}.
These approaches address popularity concentration to maximize overall matches but do not explicitly consider fairness among individual users.

\paragraph{Fair Reciprocal Recommendation}
Few studies have addressed fairness in reciprocal settings. 
\citet{xia2019we} discuss fairness among groups and similarity of mutual preference in RRSs based on the Walrasian equilibrium, which is an economic concept considering a balance between supply and demand while optimally satisfying the preferences of both parties involved. 
\citet{Virginie2021} consider fairness in RRSs based on Lorenz dominance. They define fair rankings as those with non-dominated generalized Lorenz curves for both users and items.

These models differ significantly from ours in several aspects, as summarized in Table~\ref{tab:comparison_fair_rrs}.
First, our model explicitly requires mutual interest between users for a successful match: a match occurs only when both agents apply to each other.
In contrast, \citet{xia2019we} optimize relevance-based utility without modeling the explicit mutual application process, and \citet{Virginie2021} define two-sided utility as the sum of recommendation utility and exposure without requiring mutual applications.
Second, the fairness concepts differ: we adopt envy-freeness from fair division, while \citet{xia2019we} use Walrasian equilibrium criteria (disparity of service, similarity of mutual preference, and equilibrium of demand and supply) and \citet{Virginie2021} use Lorenz dominance.
Third, \citet{Virginie2021} assume symmetric mutual preferences, whereas our model allows asymmetric preferences.

\begin{table}[t]
\centering
\caption{Comparison of fairness-aware reciprocal recommendation methods.}
\label{tab:comparison_fair_rrs}
\small
\begin{tabular}{lccc}
\toprule
 & Ours & WE-Rec~\cite{xia2019we} & Lorenz~\cite{Virginie2021} \\
\midrule
Match condition & Mutual apply & Not explicit & Not explicit \\
\midrule
Framework & Nash social welfare & Walrasian equilibrium & Welfare maximization \\
\midrule
Fairness criteria & Envy-freeness & Service disparity & Lorenz efficiency \\
 & & Mutual pref. similarity & \\
 & & Demand/supply eq. & \\
\midrule
Utility & Expected matches & Relevance-based & Two-sided utility \\
\midrule
Asymmetric pref. & \checkmark & \checkmark & $\times$ \\
\bottomrule
\end{tabular}
\end{table}

\paragraph{Fair Division in Machine Learning}
Fair division and envy-freeness have gained attention in machine learning, especially classifications and recommendations~\cite{balcan2019envyfree,golz2019paradoxes}.
In the context of binary classification, \citet{balcan2019envyfree} propose envy-freeness as a notion of fairness for classification tasks where individuals have diverse preferences over possible outcomes.
We mentioned several studies on fair recommender systems. Here, we highlight some notable work that has adopted envy-freeness to ensure fairness in recommendation settings.
\citet{patro2020fairrec} explore fairness in product recommendations to customers, considering both the utility of the customer and the exposure of the producer.
They based their fairness metrics on envy-freeness and maximin share (MMS), which is another fairness metric, and proposed fair recommendation algorithms utilizing fair division techniques.
\citet{do2022online} study fairness in user-item recommendations, focusing on envy-freeness among users. They proposed a method for certifying the fairness of the recommender system based on envy-freeness criteria.

\paragraph{Fair Division}
The problem of fair division has been extensively studied in various fields~\cite{steinhaus1948problem,Foley1967,Varian1974,Budish2011}.
It concerns allocating resources to agents in a fair and efficient manner.
In recent years, there has been active research on fair division for indivisible items, i.e., goods can not be shared between agents. In such cases, envy-free allocations may not exist and a relaxed notion of envy-freeness \emph{envy-freeness up to one good (EF1)}, has been studied~\cite{Budish2011,Lipton2004}.
An allocation is EF1 if envies can be eliminated by removing at most one good from the subset of items which an envied agent receives.
Another key concept in fair division is Nash social welfare (NSW), which aims to maximize the product of agents' utilities, striking a balance between efficiency and fairness. 
When goods are divisible and agents have additive valuations, any maximum Nash social welfare (MNW) allocation is both envy-free and \emph{Pareto optimal (PO)}~\cite{Eisenberg1959,Varian1974}, where PO ensures that no other allocation makes an agent better off without making someone else worse off. 
When goods are indivisible and agents' valuations are additive, \citet{Caragiannis2019} showed that any MNW allocation is both EF1 and PO. However, computing MNW is NP-hard~\cite{Nguyen2014}. A substantial body of research has been devoted to developing approximation algorithms aimed at maximizing the NSW~\cite{ColeGkatzelis2018,Garg2023,JainVaish2024}.
We employ these concepts of fairness and efficiency in the context of RRSs.

Recently, some papers have combined fairness concepts in fair division with two-sided matching markets~\cite{freeman2021two,gollapudi2020almost,igarashi2023fair}.
Our reciprocal recommendation problem can be considered to include the many-to-many matching problem when the examination function is limited to uniform and only deterministic recommendation list is considered. 
\citet{freeman2021two} proposed the notion of \emph{double envy-freeness up to one match (DEF1)} in two-sided matching markets, which simultaneously satisfies fairness on both sides.
They show that, when both sides have identical ordinal preferences, a complete matching satisfying DEF1 always exists and can be computed in polynomial time.
Our model assumes that agents apply to other agents with probabilities dependent on the ranking position based on the given recommendations, which makes our setting significantly different from theirs.

\section{Preliminaries}
We consider the recommendation process on two-sided matching platforms, such as those used in online dating services.
Mirroring the approach of the major online dating platforms in the real world, users on one side of the market receive recommendation lists featuring users on the opposite side. 
Based on these recommendations, users choose to apply to which users on the opposite side (or send ``likes''). 
This process is reciprocal---users on both sides receiving recommendations and making applications---leading to a match when there is mutual interest between a pair of users.

Although the primary objective of the platform is to maximize the overall number of matches, we argue that fairness among users raises important concerns.
In this section, we first describe our model of a matching platform and define a notion of fairness inspired by fair division.
For convenience, a summary of the notation used in this paper is provided in Appendix~\ref{appendix:notation}.

\subsection{Model}\label{sec:model}
In our model, the matching market is composed of two distinct sets of \emph{agents}.
For a positive integer $k$, let $[k] = \{1,2,\ldots,k\}$.
We define $N=\{a_1,a_2,\ldots,a_n\}$ as the set of $n$ agents on the left side and $M=\{b_1,b_2,\ldots,b_m\}$ as the set of $m$ agents on the right side.

For each pair of agents $(a_i, b_j)$, where $i\in [n]$ and $j\in [m]$, we introduce the estimated \emph{preference probabilities}.
Let $\hat{p}_1(i,j)\in [0,1]$ denote the estimated preference probability that agent $a_i\in N$ prefers (or would choose) agent $b_j\in M$.
Similarly, let $\hat{p}_2(j,i)\in [0,1]$ denote the estimated preference probability that agent $b_j\in M$ prefers agent $a_i\in N$.
These probabilities are typically derived from data-driven methods and are assumed to be provided in our framework.
Note that the preferences are generally asymmetric, that is, $\hat{p}_1(i,j)$ may differ from $\hat{p}_2(j,i)$ for a given pair of agents.

In the context of RRSs, each agent receives a \emph{recommendation list} of the agents on the opposite side. We employ probabilistic recommendation lists represented as doubly stochastic matrices.

For a positive integer $d$, a $d \times d$ matrix $S \in \mathbb{R}^{d \times d}_{\geq 0}$ is called a \emph{doubly stochastic matrix} if the sum of each row and each column is equal to $1$. 
In our model, each agent $a_i\in N$ receives a doubly stochastic matrix $A_i \in \mathbb{R}^{m\times m}_{\geq 0}$ representing a probabilistic ranking. We refer to $A_i$ as \emph{recommendation} for agent $a_i\in N$. The entry $(j,k)$ of $A_i$, denoted as $A_i(j,k)$, represents the probability that agent $b_j\in M$ is placed in the $k$-th position in the ranking of agent $a_i$.
Similarly, agent $b_j\in M$ receives a doubly stochastic matrix $B_j \in \mathbb{R}^{n\times n}_{\geq 0}$ and we call $B_j$ the recommendation for agent $b_j\in M$. 
We denote $\bm{A}=(A_1,A_2,\ldots,A_n)$ and $\bm{B}=(B_1,B_2,\ldots,B_m)$, and define \emph{recommendation policy} as the pair $(\bm{A},\bm{B})$, which encompasses the recommendations for all agents. 

\subsection{Utilities and Social Welfare}\label{sec:social_welfare}
We first consider recommendation policies that aim to maximize the sum of utilities of all agents on both sides.
Our primary objective is to develop a framework for identifying ``optimal'' recommendation policies given the estimated preference probabilities.

We now explain how agents apply to agents on the opposite side based on the recommendations they receive.
We adopt the position-based model (PBM), a commonly assumed model in recommendation contexts~\cite{craswell2008experimental,Joachims2017-ix}, to determine the probability that an agent applies to another agent.
Let $e(k)$ denote the examination probability, representing the likelihood that an agent will examine the $k$-th candidate in their recommendation list.
The function $e$ is assumed to be non-increasing with respect to $k$, reflecting the observation that agents are less likely to consider candidates lower in the list. 
The specific definition of $e$ depends on the context of the real-world application.
Common examples of examination probability functions in the literature include $e(k) = \frac{\mathbbm{I}(k\leq K)}{k} $ or $\frac{\mathbbm{I}(k\leq K)}{\log_2 (k+1)}$, where $K \geq 0$ represents a threshold number\footnote{Agents will not apply to agents who are beyond the $K$-th position in the ranking, or they may not receive recommendations for such agents. For an event $\mathcal{E}$, $\mathbbm{I}(\mathcal{E})$ is an indicator function such that $\mathbbm{I}(\mathcal{E})=1$ if $\mathcal{E}$ holds and $0$ otherwise.}.
These functions are widely used~\cite{jarvelin2002cumulated,SaitoJoachimsKDD2022}.
Based on PBM, we suppose that
the probability that an agent $a_i \in N$ applies to an agent $b_j\in M$ on the opposite side is given by the product of $a_i$'s estimated preference for $b_j$ and the expected value of $e(k)$ where $k$ is the position of $b_j$ in the recommendation list for $a_i$.
Thus, given recommendation $A_i$, the probability that agent $a_i\in N$ applies to agent $b_j\in M$ is 
$$
    \mathrm{Pr}[\text{$a_i$ applies to $b_j$}] 
    = \hat{p}_1(i,j)\sum_{k=1}^{m} e(k)  \cdot A_i(j,k),
$$  
and similarly, the probability that agent $b_j\in M$ applies to agent $a_i\in N$ is
$$
    \mathrm{Pr}[\text{$b_j$ applies to $a_i$} ] 
    = \hat{p}_2(j,i)\sum_{\ell=1}^{n} e(\ell)  \cdot B_j(i,\ell).
$$

In our model, a match between two agents $a_i\in N$ and $b_j \in M$ is established only when both $a_i$ applies to $b_j$ and $b_j$ applies to $a_i$.
Thus, the probability that the agent $a_i\in N$ matches with the agent $b_j \in M$ is given by
\begin{align*}
    &\mathrm{Pr}[\text{$a_i$ matches with $b_j$} ] 
    = 
        \mathrm{Pr}[\text{$a_i$ applies to $b_j$} ~\land~ \text{$b_j$ applies to $a_i$}] \\
    &= 
        \mathrm{Pr}[\text{$a_i$ applies to $b_j$} ] \cdot \mathrm{Pr}[\text{$b_j$ applies to $a_i$} ] 
    = 
        p_{ij} \sum_{k=1}^m\sum_{\ell=1}^{n} e_{k\ell}  A_{i}(j,k) \cdot B_{j}(i,\ell),
\end{align*}
where $p_{ij} = \hat{p}_1(i,j)\cdot \hat{p}_2(j,i)$ and $e_{k\ell} = e(k)\cdot e(\ell)$ for each $k \in [m]$ and $\ell \in [n]$.

Fig.~\ref{fig:model} illustrates this conceptual model. As shown, agents on both sides receive recommendations in the form of doubly stochastic matrices that encode the probabilities of ranking candidates at various positions.
\begin{figure*}[t]
    \begin{center}
    \includegraphics[width=0.8\linewidth]{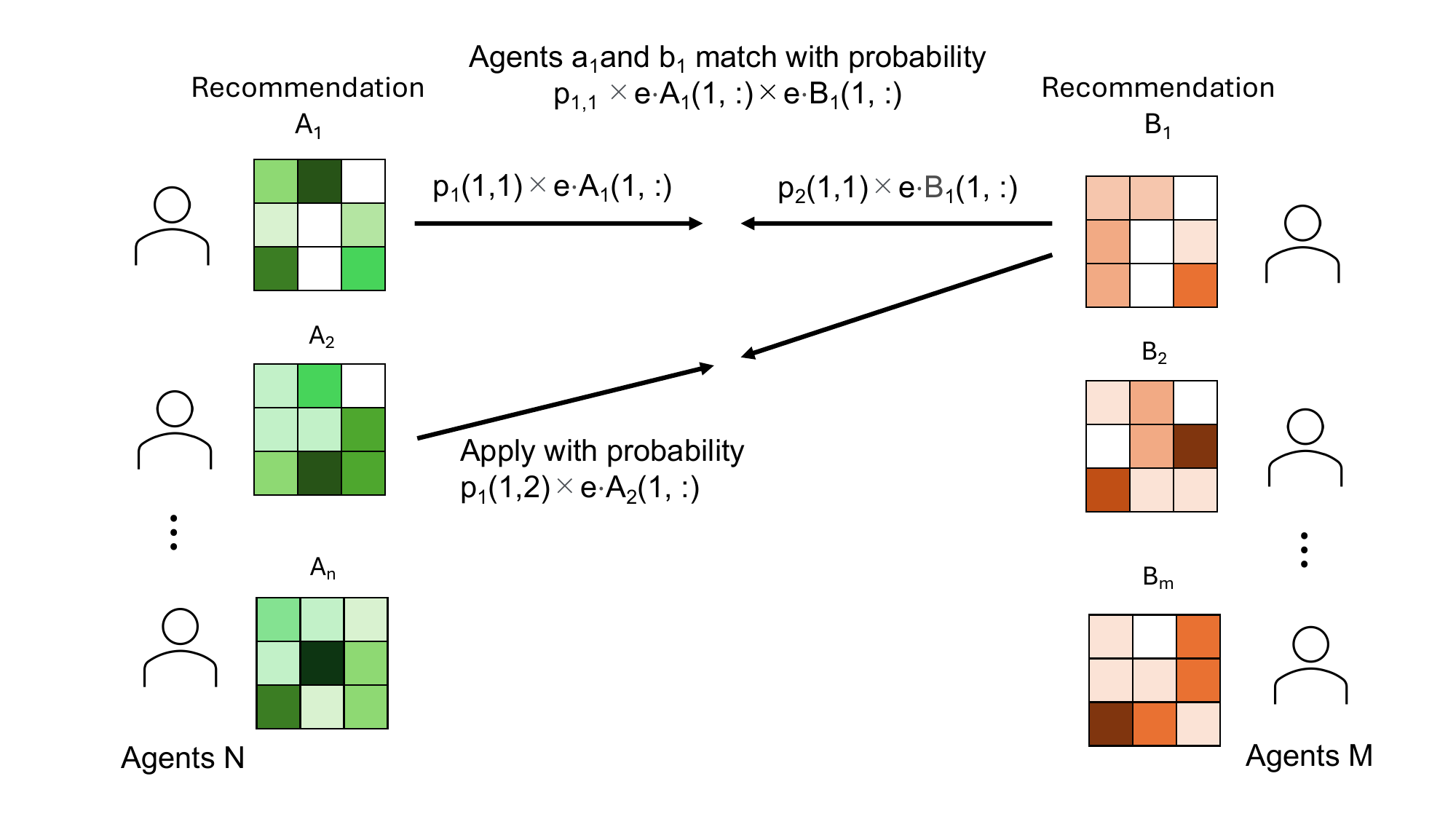}
    \caption{Overview of the recommendation process in two-sided matching platforms. Agents on both sides receive probabilistic recommendation lists---represented by doubly stochastic matrices---and make application decisions based on estimated preference probabilities. A match is established when both agents mutually select each other.}
    \label{fig:model}
    \end{center}
\end{figure*}

In our two-sided recommendation setting, we assume that each agent obtains \emph{utility} defined as the expected number of successful matches with other agents on the opposite side, based on the recommendation policy.
Note that an individual may match with more than one partner.
Specifically, under policy $(\bm{A},\bm{B})$, the utility of agent $a_i\in N$ is given by 
$$
U_{i}(\bm{A},\bm{B}) = \sum_{b\in M} \mathrm{Pr}[\text{$a_i$ matches with $b$}]
= \sum_{j=1}^m
p_{ij} \sum_{k=1}^m\sum_{\ell=1}^{n} e_{k\ell}  A_{i}(j,k)  B_{j}(i,\ell)
.
$$ 
Note that the utility $U_{i}(\bm{A},\bm{B})$ depends only on $A_i$ and $\bm{B}$.
Similarly, the utility of agent $b_j\in M$ under policy $(\bm{A},\bm{B})$ is given by 
$$
V_{j}(\bm{A},\bm{B}) = \sum_{a\in N} \mathrm{Pr}[\text{$b_j$ matches with $a$}]
= \sum_{i=1}^n
p_{ij} \sum_{k=1}^m\sum_{\ell=1}^{n} e_{k\ell}  A_{i}(j,k)  B_{j}(i,\ell)
.
$$

We now define the \emph{social welfare} $\mathrm{SW}(\bm{A},\bm{B})$ for policy $(\bm{A},\bm{B})$ as the expected number of matches between all agents on the two sides, that is,
$$
\mathrm{SW}(\bm{A},\bm{B}) 
    = \sum_{i=1}^n U_i(\bm{A},\bm{B})
    = \sum_{j=1}^m V_j(\bm{A},\bm{B}).
$$
A policy $(\bm{A},\bm{B})$ is said to be \emph{socially optimal} if it maximizes social welfare.

If an agent is aware of the recommendations provided to all agents on the opposite side, the agent would desire the most favorable recommendation to maximize its own utility.
For each agent $a_i \in N$, we say that the recommendation $A_i^*$ is \emph{best} for given $\bm{B}$ if $A_i^* \in \mathrm{argmax}_{A_i}U_i (A_{i},\bm{B})$, under the constraint that $A_i$ is a doubly stochastic matrix.
Similarly, for each agent $b_j\in M$, the recommendation $B_j^*$ is best for given $\bm{A}$ if $B_j^* = \mathrm{argmax}_{B_j}V_j (\bm{A},B_{j})$. 
It is straightforward to verify that if the recommendation policy is socially optimal, then all agents receive their best recommendations.

\subsection{Fairness Axioms}\label{sec:fairness_axioms}
We introduce fairness notions for recommendation policies in our reciprocal recommendation model.

We first define the concept of the \emph{opportunity} to be recommended, inspired by the concept of the \emph{impact} defined by \citet{SaitoJoachimsKDD2022}\footnote{\citet{SaitoJoachimsKDD2022} defined the concept of impacts for items that are recommended to users, and introduce envy-freeness among items with respect to impacts. 
If we map agents in $N$ to users and agents in $M$ to items and fix $\bm{A}$, then the opportunity of agent $j$ in our model coincides precisely with the impact of item $j$.}.
The opportunity for an agent represents all the rows relevant to that agent from every recommendation list provided on the opposite side.
Specifically, the opportunity of the agent $a_i\in N$ is represented by a tuple of
$(B_1(i,:), B_2(i,:), \ldots, B_m(i,:))$
where $B_j(i,:)\in \mathbb{R}^{n}$ is the $i$-th row of the recommendation matrix $B_j$ for each $j \in [m]$. 
Similarly, the opportunity of agent $b_j \in M$ is given by
$(A_1(j,:), A_2(j,:), \ldots, A_n(j,:))$ 
where $A_i(j,:)\in \mathbb{R}^{m}$ is the $j$-th row of the recommendation matrix $A_i$ for each $i \in [n]$.


Next, we define a concept of fairness for the recommendation policy that ensures that the opportunities to be recommended are "fair" for all agents from the perspective of fair division~\cite{Foley1967,Varian1974}.
For example, consider a scenario where users can assess the probability of being recommended to the other side. 
If two users on the same side have identical situations (e.g., similar features or preference levels from the opposite side), but one user is clearly being recommended more favorably to the other side than the other user, the disadvantaged user may feel envious and dissatisfied with the platform.

To capture such states of unfairness, we introduce the concept of \emph{envy-freeness} with respect to opportunity.
An agent $a_i\in N$ is said to \emph{envy} another agent $a_{i'}\in N$ in terms of opportunity if the utility that agent $a_i$ would receive from $a_{i'}$'s opportunity exceeds the utility that $a_i$ receives from their own opportunity.
Let $U_i\bigl(B_1(i',:), B_2(i',:), \ldots, B_m(i',:)\bigr)$ be the utility that agent $a_i$ receives from $a_{i'}$'s opportunity, i.e.,
\[
    U_i\bigl(B_1(i',:), B_2(i',:), \ldots, B_m(i',:)\bigr) = \sum_{j=1}^m \sum_{k=1}^{m} \sum_{\ell=1}^{n} p_{ij} e_{k\ell} A_{i}(j,k)\cdot B_{j}(i',\ell).
\]
Agent $a_i$ envies agent $a_{i'}$ if and only if $U_i\bigl(B_1(i,:), B_2(i,:), \ldots, B_m(i,:)\bigr) < U_i\bigl(B_1(i',:), B_2(i',:), \ldots, B_m(i',:)\bigr)$.

The situation is analogous for the agents in $M$. 
We define the utility that agent $b_j \in M$ receives from agent $b_{j'}$'s opportunity as
\[
    V_j\bigl(A_1(j',:), A_2(j',:), \ldots, A_n(j',:)\bigr) = \sum_{i=1}^n \sum_{k=1}^{m} \sum_{\ell=1}^{n} p_{ij} e_{k\ell} A_{i}(j',k)\cdot B_{j}(i,\ell).
\]
We say that agent $b_j$ envies agent $b_{j'}$ if $V_j\bigl(A_1(j,:), A_2(j,:), \ldots, A_n(j,:)\bigr) 
< V_j\bigl(A_1(j',:), A_2(j',:), \ldots, A_n(j',:)\bigr)$.


A policy is \emph{doubly envy-free} if no agent envies the recommendation opportunities of any other agent on the same side of the market. We formalize this concept as follows:
\begin{definition}[Double envy-freeness]
    Recommendations $B_1,B_2,\ldots,B_n$ is said to be
    \emph{left-side envy-free} if for every pair of agents $a_{i},a_{i'} \in N$,
    $
    U_i(B_1(i,:), B_2(i,:), \ldots, B_m(i,:))
    \ge 
    U_i(B_1(i',:), B_2(i',:), \ldots, B_m(i',:))
    $
    Similarly, recommendations $A_1,A_2,\ldots,A_m$ is said to be \emph{right-side envy-free} if for every pair of agents $b_{j},b_{j'}\in M$, 
    $
    V_j\bigl(A_1(j,:), A_2(j,:), \ldots, A_n(j,:)\bigr) 
    \ge 
    V_j\bigl(A_1(j',:), A_2(j',:), \ldots, A_n(j',:)\bigr).
    $
    A recommendation policy $(\A,\B)$ is called \emph{double envy-free} if its recommendations of both sides simultaneously satisfy envy-freeness on the left side and on the right side.
\end{definition}
It is worth noting that a double envy-free policy always exists in RRSs. 
To illustrate this, consider a scenario in which we assign probabilistic recommendation lists such that all agents have the same probability of appearing at any position on the lists of agents on the opposite side.  Formally, consider the uniform recommendation policy $ (\bm{A}^{\mathrm{uni}},\bm{B}^{\mathrm{uni}})$, where
$A_i^{\mathrm{uni}}(j,k)=\frac{1}{m}$ for all $i, j, k$ and
$B_j^{\mathrm{uni}}(i,\ell)=\frac{1}{n}$ for all $j,i,\ell$.
%
Although this policy ensures double envy-freeness, it may not maximize the number of successful matches.

%
In contrast, maximizing social welfare does not guarantee that the resulting policy will be double envy-free. 
The presence of conflicting preferences among agents can lead to a trade-off between the overall utility and fairness for all participants.
Consider the following example.
\begin{example}\label{example:non-envy-free}
    Consider a simple scenario with two agents and one agent. 
    Thus, $n=2$ and $m=1$.
    Let $\varepsilon\in (0,1)$ be a constant. 
    Assume that the inverse examination function is given by $e(k)=1/k$, and the preference probabilities are given by $\hat{p}_1 = (1, 1)^{\top} \in \mathbb{R}^{2\times 1}$, and $\hat{p}_2 = (1, 1-\varepsilon) \in \mathbb{R}^{1\times 2}$.
    This means that the single agent on $M$ equally prefers both agents on $N$, and that while both agents on $N$ strongly prefer the agent on $M$, the preference for $a_1$ is slightly higher than $a_2$'s.
    
    For this instance, 
    we can see that the policy $A_1=1,A_2=1,B_1=[[1,0],[0,1]]$ is socially optimal.
    Here, the components $A_1 = A_2 = 1$ mean that the single agent in $M$ is recommended to the two agents with probability $1$.
    $B_1 = [[1,0],[0,1]]$ represents the recommendation matrix where $a_1$ is always recommended in the first position to the recommendation list of the agent on $M$, and $a_2$ is always recommended in the second position in the list.
    Although this policy achieves the highest social welfare, it falls in terms of fairness. Specifically, it's not double envy-free since we have $U_2(B_1(2,:)) = U_2([0,1]) = \frac{1-\varepsilon}{2} <  1-\varepsilon = U_2(B_1(1,:))$, and agent $a_2$ will envy agent $a_1$.
    
    On the other hand, the uniform recommendation policy $A^{\mathrm{uni}}_1=1,A^{\mathrm{uni}}_2=1,B^{\mathrm{uni}}_1=[[1/2,1/2],[1/2,1/2]]$ satisfies double envy-freeness. However, this uniform policy is not socially optimal, since $\mathrm{SW}(\bm{A}^{\mathrm{uni}}, \bm{B}^{\mathrm{uni}}) = \frac{3}{4} + \frac{3(1-\varepsilon)}{4} < 1+ \frac{1-\varepsilon}{2} = \mathrm{SW}(\bm{A},\bm{B})$.
\end{example}

\subsection{Relationship to Fair Division}
In this section, we explore the connection between the fairness axioms introduced in our reciprocal recommendation model and the well-established concepts in fair division.
Namely, we present that the problem of finding a left-side envy-free recommendation in a special case is equivalent to the problem of finding an envy-free allocation in fair division for divisible items.
Specifically, when $K=1$ and recommendations $A_1,A_2,\ldots,A_n$ are fixed, 
achieving left-side envy-freeness requires finding opportunities for all agents in $N$, represented as $B_1(i,1), B_2(i,1), \ldots, B_m(i,1)$, such that
$($i$)$ for all pairs of agents $a_{i},a_{i'}\in N$, 
$\sum_{j\in M}p_{ij}A_i(j,1) \cdot B_j(i,1) \geq \sum_{j\in M}p_{ij}A_i(j,1) \cdot B_{j}(i',1)$,
$($ii$)$ $\sum_{i=1}^n B_j(i,1) = 1$ for each $j \in [m]$,
and $($iii$)$ $B_j(i,1) \ge 0$ for all $i\in [n], j\in [m]$.
In this case, $B_1(i,1), B_2(i,1), \ldots, B_m(i,1)$ does indeed represent an envy-free allocation of left-side agents to right-side agents, where each agent $a_i\in N$ has a utility of $p_{ij}A_i(j,1)$ for each agent $b_j\in M$.

\subsection{Pareto Optimality of Recommendation Opportunity}\label{sec:Pareto}
While envy-freeness captures fairness among agents, it does not guarantee that the recommendation policy is efficient.
For example, the uniform recommendation policy is always double envy-free, but it may fail to utilize recommendation opportunities effectively.
To assess efficiency, we introduce the concept of Pareto optimality.

In economics, an allocation is Pareto optimal if no agent can be made better off without making another agent worse off.
In our RRS model, the opportunity to be recommended can be viewed as a resource to be allocated among agents.
Intuitively, a recommendation policy is Pareto optimal if the platform cannot increase the expected matches that one agent obtains from recommendations without decreasing those of another agent.

We formally define Pareto optimality for recommendations in our model as follows:
\begin{definition}[Pareto optimality]
    Recommendations $\bm{B}=(B_1,\ldots,B_n)$ is \emph{Pareto dominated} for the left side given $\bm{A}$ by the other $\bm{B}'=(B_1',\ldots,B_m')$ if $U_i(A_i, \bm{B}') \ge U_i(A_i, \bm{B})$ for all agents on the left side $a_i \in N$, with strict inequality for at least one agent.
    Furthermore, $\bm{B}$ is \emph{Pareto optimal} for the left side given $\bm{A}$ if there exists no other $\bm{B}'$ that Pareto dominates for the left side.
    The definition extends analogously to the right side.
    We say that a recommendation policy $(\bm{A}, \bm{B})$ is \emph{mutually Pareto optimal} if $\bm{A}$ is Pareto optimal given $\bm{B}$ and $\bm{B}$ is Pareto optimal given $\bm{A}$.
\end{definition}

In Example~\ref{example:non-envy-free}, we demonstrate that while socially optimal recommendations may not be double envy-free, the uniform recommendation policy---which suggests agents uniformly to the other side (i.e., $A_{i}(j, k) = 1/m$ and $B_{j}(i, \ell) = 1/n$ for all $i,j,k,\ell$)---is double envy-free.
However, as the following example illustrates, the uniform recommendation policy may not be Pareto optimal.

\begin{example}
Consider a scenario with two agents on the left side $N = \{a_1, a_2\}$, one agent on the right side $M = \{b_1\}$, and $e(1) = 1, e(2) = 0$.
The preference probabilities are $\hat{p}_{1}(1,1) = 1, \hat{p}_{1}(2,1) = 0, \hat{p}_2(1,1) = \hat{p}_2(1,2) = 1$.
In this setting, only $a_1$ wants to match $b_1$ while $a_2$ does not, and the agent on the right side $b_1$ is indifferent between the two agents on the left side.
Under the uniform recommendation policy $(\bm{A}^\mathrm{uni}, \bm{B}^\mathrm{uni})$, which recommends both agents on the left side with the same probability to $b_1$, we have $U_{1}(A_1^\mathrm{uni}, \bm{B}^\mathrm{uni}) = 1/2$ and $U_{2}(A_2^\mathrm{uni}, \bm{B}^\mathrm{uni}) = 0$.
Let $\bm{B}' = [[1,0],[0,1]]$, which recommends $a_1$ to $b_1$ in the first position with probability 1, then the utilities for the agents on the left side are $U_{1}(A_1^{\mathrm{uni}}, \bm{B}') = 1$ and $U_{2}(A_2^{\mathrm{uni}}, \bm{B}') = 0$.
Therefore, $\bm{B}^{\mathrm{uni}}$ is Pareto dominated by $\bm{B}'$ for the left side, demonstrating that the uniform recommendation policy is not Pareto optimal in this case.
\end{example}

Pareto optimality is closely related to social welfare maximization defined in Section~\ref{sec:social_welfare}. The following proposition shows that any socially optimal policy is also Pareto optimal.

\begin{proposition} 
    If a policy is socially optimal, then it is mutually Pareto optimal.
\end{proposition}
\begin{proof}
    Let $(A^*, B^*)$ be a socially optimal policy.
    Suppose for contradiction that \(B^*\) is not Pareto optimal for the left side given \(A^*\).
    Then, there exists another $B'$ such that $U_{i}(A^*, B') > U_{i}(A^*, B^*)$ for some $a_i \in N$, and $U_{i'}(A^*, B') \ge U_{i'}(A^*, B^*)$ for all $a_{i'} \in N$.
    Thus,
    \begin{align*}
        \mathrm{SW}(A^*, B^*) 
        = U_{i}(A^*, B^*) 
        + \sum_{i'\in [n]\setminus\{i\}}U_{i'}(A^*,B^*) 
        <  U_{i}(A^*, B') 
        + \sum_{i'\in [n]\setminus\{i\}}U_{i'}(A^*,B') 
        = \mathrm{SW}(A^*, B'),
    \end{align*}
    which contradicts the assumption that $(A^*, B^*)$ is socially optimal.
    The proof for the right side is analogous.
\end{proof}

\section{Computing Socially Optimal Policy}\label{sec:SW}

In this section, we propose a reciprocal recommendation algorithm which computes a socially optimal policy approximately.

By utilizing $n^2m+nm^2$ variables $\{A_{i}(j,k)\}_{i,j,k}$ and $\{B_{j}(i,\ell)\}_{j,i,\ell}$,
our optimization problem of finding a socially optimal policy 
is to maximize $\mathrm{SW}(\bm{A},\bm{B})$ under the constraint that $A_i$ and $B_j$ are doubly stochastic matrices for all $i\in [n]$ and all $j\in [m]$.
Since the optimization objective is quadratic, computing the global optimal solution is challenging.

Although the function $\mathrm{SW}(\bm{A},\bm{B})$ is not necessarily concave with respect to $(\bm{A},\bm{B})$, it is readily apparent that fixing $\bm{A}$ makes $\mathrm{SW}$ a linear function of $\bm{B}$ and vice versa.
As a naive approach, one might naturally conceive of an alternating maximization method: first optimizing over $\bm{A}$ while keeping $\bm{B}$ fixed and then optimizing over $\bm{B}$ with $\bm{A}$ fixed.
This procedure essentially computes the best recommendations for each side alternately at each step. If the process converges to some values $(\bm{A}^*,\bm{B}^*)=(A_1^*,A_2^*,\ldots,A_n^*,B_1^*,B_2^*,\ldots,B_m^*)$, then $A_i^* \in \mathrm{argmax}_{A_i}U_i (A_{i},\bm{B}^*)$ and
$B_j^* \in \mathrm{argmax}_{B_j}V_j (\bm{A}^*,B_{j})$
for all $i\in N$ and for all $j\in M$.
From the linearity of the max function, we also have 
$\bm{A}^* \in \mathrm{argmax}_{\bm{A}}\;\mathrm{SW} (\bm{A},\bm{B}^*)$
and 
$\bm{B}^* \in \mathrm{argmax}_{\bm{B}}\;\mathrm{SW} (\bm{A}^*,\bm{B})$.


We now present a general algorithm for alternating optimization, as described in Algorithm~\ref{alg:alternating_SW_maximization}, which employs the Frank-Wolfe algorithm to solve the convex optimization problems~\cite{Frank1956,Su2022}. 
This framework will also be utilized in subsequent sections. 
Consider two real-valued functions, $F_1$ and $F_2$, which take matrices $\bm{A}$ and $\bm{B}$ as input. 
We assume that the function $F_1(\bm{A},\bm{B})$ is concave with respect to $\bm{B}$, and $F_2(\bm{A},\bm{B})$ is concave with respect to $\bm{A}$. 
Consequently, when $\bm{A}$ (resp. $\bm{B}$) is fixed, $F_1(\bm{A},\bm{B})$ (resp. $F_1(\bm{A},\bm{B})$) can be maximized with respect to $\bm{B}$ (resp. $\bm{A}$) using the Frank-Wolfe algorithm.
Algorithm~\ref{alg:alternating_SW_maximization} begins by initializing doubly stochastic matrices $\bm{A}$ and $\bm{B}$ in line~\ref{line:initialize_A_SW}. 
The alternating optimization process is then executed within the \textbf{for} loop (lines~\ref{line:loop_start_A_SW}), continuing until convergence is achieved or a predefined number of iterations is completed.

\begin{algorithm}[t]
    \caption{Alternating Maximization of Functions $F_1(\bm{A},\bm{B})$ and $F_2(\bm{A},\bm{B})$  via Frank-Wolfe Algorithm}
    \label{alg:alternating_SW_maximization}
    \begin{algorithmic}[1]
        \REQUIRE Preference probabilities $(p_{ij})_{i,j}$, an examination function $e(\cdot)$, and learning rates $(\eta_t)_{t\in [T]}$.
        \ENSURE Recommendations $\bm{A},\bm{B}$.
        \STATE Initialize $\bm{A},\bm{B}$.\label{line:initialize_A_SW}
        \FOR{$t=1,2,\ldots,T$} \label{line:loop_start_A_SW}
            \STATE 
                $\bm{X}^* \in \text{argmax}_{\bm{X} = (X_{i})_{i\in N}} \sum_{i=1}^n \sum_{j,k=1}^m \left(\nabla_{A_i} F_2(\bm{A},\bm{B})\right)_{j,k} \cdot \left(X_{i} \right)_{j,k}$ s.t. $X_{i}\boldsymbol{1} =\boldsymbol{1}$, $\boldsymbol{1}^{\top}X_{i} = \boldsymbol{1}^{\top}$, and $X_{i} \in \mathbb{R}^{m\times m}_{\geq 0}$ $\forall i\in [n]$.\label{line:maximize_A_SW}
            \STATE 
                $\bm{A} \leftarrow (1-\eta_t)\bm{A}  + \eta_t \bm{X}^*$
            \STATE 
                $\bm{Y}^* \in \text{argmax}_{\bm{Y} = (Y_{j})_{j\in M}} \sum_{j=1}^m \sum_{i,\ell=1}^n \left(\nabla_{B_j} F_1(\bm{A},\bm{B}) \right)_{i,\ell} \cdot \left( Y_j \right)_{i,\ell} $ s.t. $Y_j\boldsymbol{1} =\boldsymbol{1}$, $\boldsymbol{1}^{\top}Y_j = \boldsymbol{1}^{\top}$, and $Y_j \in \mathbb{R}^{n\times n}_{\geq 0}$ $\forall j\in [m]$.\label{line:maximize_B_SW}
            \STATE 
                $\bm{B} \leftarrow (1-\eta_t)\bm{B}  + \eta_t \bm{Y}^*$
        \ENDFOR \label{line:while_end_SW}
        \RETURN $\bm{A},\bm{B}$.
    \end{algorithmic}
\end{algorithm}

To apply this framework to alternate maximization of social welfare, we set $F_1(\bm{A},\bm{B})=F_2(\bm{A},\bm{B})=\mathrm{SW}(\bm{A},\bm{B})$.
If the values of $\bm{A}$ and $\bm{B}$ converge in the alternating algorithm, we obtain a policy ensuring that all agents receive their best recommendations.
Although this procedure may be anticipated to converge to a local optimum or fail to converge, the experiments described in Section~\ref{sec:experiments} demonstrate convergence and the ability to achieve high social welfare values.
For each $i\in [n]$, $j \in [m]$, $k\in [m]$, and $\ell\in [n]$,
the gradient calculations are as follows
\begin{align*}
    \left( \nabla_{A_i} \mathrm{SW}(\A,\B)\right)_{j,k} =   \frac{\partial}{\partial A_i(j,k)} \mathrm{SW}(\A,\B) 
    = p_{ij}  \sum_{\ell=1}^{n} e_{k\ell} B_{j}(i,\ell),
\end{align*}
and
\begin{align*}
    \left( \nabla_{B_j} \mathrm{SW}(\A,\B)\right)_{i,\ell} 
    = \frac{\partial}{\partial B_j(i,\ell)} \mathrm{SW}(\A,\B) 
    =  p_{ij}  \sum_{k=1}^{m} e_{k\ell} A_i(j,k).
\end{align*}

Although a socially optimal policy $(A^*, B^*)$ is not necessarily envy-free, as shown in Example~\ref{example:non-envy-free}, it is guaranteed that the policy is mutually Pareto optimal.

\section{Computing Fair Recommendation}
\label{sec:NSW}

As illustrated in Example~\ref{example:non-envy-free}, socially optimal policies do not always guarantee double envy-freeness.
This observation is further supported by the experimental results in Section~\ref{sec:experiments}, which reveal that the policies derived from Algorithm~\ref{alg:alternating_SW_maximization} and other algorithms designed to maximize social welfare can result in significant levels of envy across both sides of the market.
We present a novel approach that aims to achieve fairness in the opportunities to be recommended among agents in RRSs. 

\subsection{Nash Social Welfare}

We propose a method that aims to compute a double envy-free policy using the \emph{Nash social welfare (NSW)} functions~\cite{Eisenberg1959,Varian1974,Kroer2019}. 
This approach is an extension of the work by Saito and Joachims~\cite{SaitoJoachimsKDD2022}, who considered the problem of recommending items to users in a one-sided recommendation setting. 
We adapt their technique to the reciprocal recommendation problem, where both sides of the market receive recommendations.

We introduce two NSW-style functions.
For policy $(\bm{A},\bm{B})$, the \emph{left-side NSW} function $\mathrm{NSW}_1(\bm{A}, \bm{B})$ is defined as the product of the expected utilities that each agent in $N$ receives under the policy. Formally,
\begin{align*}
    \mathrm{NSW}_1(\bm{A}, \bm{B}) 
    = \prod_{i=1}^n U_i(\bm{A}, \bm{B}) 
    = \prod_{i=1}^n \sum_{j=1}^m p_{ij}  \sum_{k=1}^{m} \sum_{\ell=1}^{n} e_{k\ell} A_{i}(j,k) B_{j}(i,\ell).
\end{align*}
Similarly, the \emph{right-side NSW} function $\mathrm{NSW}_2(\bm{A}, \bm{B})$ for policy $(\bm{A},\bm{B})$ is given by
\begin{align*}
    \mathrm{NSW}_2(\bm{A},\bm{B}) 
    = \prod_{j=1}^m V_j(\bm{A}, \bm{B})
    = \prod_{j=1}^m \sum_{i=1}^n p_{ij}  \sum_{k=1}^{m} \sum_{\ell=1}^{n} e_{k\ell} A_{i}(j,k) B_{j}(i,\ell).
\end{align*}

Since NSW functions on the left and right sides are defined as products of individual utilities, maximizing these functions naturally favors balanced allocations for each side.
%
We say that $\bm{B}^*=(B_1^*,B_2^*,\ldots,B_m^*)$ is \emph{NSW-best} for the left side given $\bm{A}$ if $\bm{B}^*$ maximizes $\mathrm{NSW}_1(\bm{A}, \bm{B})$.
Similarly, $\bm{A}^*=(A_1^*,A_2^*,\ldots,A_n^*)$ is \emph{NSW-best} for the right side given $\bm{B}$ if $\bm{A}^*$ maximizes $\mathrm{NSW}_2 (\bm{A}, \bm{B})$.
We say that $(\bm{A}^*,\bm{B}^*)$ is mutually \emph{NSW-best} if both of $\bm{A}^*$ and $\bm{B}^*$ are NSW-best for the other.

Analogously to the SW-best policy, the NSW-best policy is Pareto optimal.
\begin{proposition}
    Suppose that $(\A^*, \B^*)$ is mutually NSW-best and all agents obtain positive utilities, i.e., $U_{i}(\A^*,\B^*)>0~\forall i \in [n]$ and $V_{j}(\A^*,\B^*)>0~\forall j \in [m]$. Then $(A^*, B^*)$ is mutually Pareto optimal.
\end{proposition}
\begin{proof}
    By symmetry, it suffices to prove for left-side agents. 
    Suppose, for contradiction, that $\B^*$ is not Pareto optimal for the left side given $\A^*$.
    Then there exists $\B'$ such that $U_{i}(\A^*, \B') > U_{i}(\A^*, \B^*)$ for some agent $a_i \in N$ and $U_{i'}(\A^*, \B') \ge U_{i'}(\A^*, \B^*)$ for all $a_{i'} \in N$. This implies
    \begin{align*}
        \mathrm{NSW}_1(\A^*, \B^*) 
        = 
            U_{i}(\A^*, \B^*) \cdot \prod_{i'\in [n]\setminus\{i\}}U_{i'}(\A^*,\B^*) 
        <  
            U_{i}(\A^*, \B') \cdot \prod_{i'\in [n]\setminus\{i\}}U_{i'}(\A^*,\B') 
        = 
            \mathrm{NSW}_1(\A^*, \B'),
    \end{align*}
    which contradicts the optimality of $\B^*$ in maximizing \(\mathrm{NSW}_1\).
\end{proof}

Theorem 4.1 of \citet{SaitoJoachimsKDD2022} and Theorem 1 of \citet{Kroer2019} immediately imply that if two sets of recommendations simultaneously satisfy NSW-best under a certain condition, then the policy is approximately double envy-free.

To elucidate this, we introduce a definition pertaining to the preferences of agents. 
The values of preference probabilities $(p_{ij})_{i,j}$ satisfies $\varepsilon$-\emph{similarity} for a nonnegative constant $\varepsilon\geq 0$ if for every left-side agent $a_i\in N$, there exists a set of the left-side agents $N'\subseteq N$ of size at least $K+1$ where each left-side agent $a_{i'}\in N'$ satisfies $\max_{j\in [m]}|p_{ij} - p_{i'j}|\leq \varepsilon$, and, simultaneously, for every right-side agent $b_j\in M$, there exists a subset $M'\subseteq M$ of size at least $K+1$ where each right-side agent $b_{j'}\in M'$ satisfies $\max_{i\in [n]}|p_{ij} - p_{ij'}|\leq \varepsilon$.
When $n$ and $m$ are substantially larger than $K$, the assumption that there is at most a constant number $K+1$ of similar individuals for each agent in the market does not impose an excessively strict constraint.
Then, we present the following theorem.
\begin{theorem}\label{thm:envy-free-policy}
    For the case of $K=1$, if a policy induces NSW-best recommendations for both sides, then the policy is doubly envy-free.
    %
    For the case of $K>1$, letting $\varepsilon$ be a non-negative constant, we assume that the estimated preference probabilities satisfies $\varepsilon$-similarity, and $n=\Theta(m)$\footnote{I.e., there exist two positive constants $c_1$ and $c_2$ such that $c_1\cdot m \leq n \leq c_2\cdot m$.}. If a policy $(\bm{A},\bm{B})$ induces NSW-best recommendations for both sides,
    then the policy satisfies double envy-freeness up to an additive difference of $\mathrm{O}(\varepsilon)$, that is, for every pair of the left-side agents $a_{i},a_{i'}\in N$,
    $$
    U_i(B_1(i,:), B_2(i,:), \ldots, B_m(i,:))
    \ge 
    U_i(B_1(i',:), B_2(i',:), \ldots, B_m(i',:)) - \mathrm{O}(\varepsilon),
    $$
    and
    for every pair of the right-side agents $b_{j},b_{j'}\in M$,
    $$
    V_j\bigl(A_1(j,:), A_2(j,:), \ldots, A_n(j,:)\bigr) 
    \ge 
    V_j\bigl(A_1(j',:), A_2(j',:), \ldots, A_n(j',:)\bigr)- \mathrm{O}(\varepsilon).
    $$
\end{theorem}

\subsection{Alternate maximization of NSW}
\label{sec:alternate-maximization-of-NSW}

We now discuss a method for finding NSW-best recommendations for both sides. We employ the alternating optimization technique described in Algorithm~\ref{alg:alternating_SW_maximization}.

When $\bm{A}$ is fixed, $\log \mathrm{NSW}_2(\bm{A},\bm{B})$ is a concave function over $\bm{B}$ since it is a sum of logarithms of sums. Similarly, when $\bm{B}$ is fixed, $\log \mathrm{NSW}_1(\bm{A},\bm{B})$ is a concave function over $\bm{A}$.
Therefore, we can apply the alternating maximization algorithm
by setting $F_1(\bm{A},\bm{B}) = \log \mathrm{NSW}_1(\bm{A},\bm{B})$ and $F_2(\bm{A},\bm{B}) = \log \mathrm{NSW}_2(\bm{A},\bm{B})$ in Algorithm~\ref{alg:alternating_SW_maximization}.
The gradient calculations are given by 
\begin{align*}
    \left( \nabla_{A_i} \log \mathrm{NSW}_2(\A,\B)\right)_{j,k} 
    = \frac{\partial}{\partial A_i(j,k)} \log \mathrm{NSW}_2(\bm{A},\bm{B}) 
    = \frac{1}{V_j(\pi)} p_{ij}  \sum_{\ell=1}^{n} e_{k\ell} B_{j}(i,\ell),
\end{align*}
and
\begin{align*}
    \left( \nabla_{B_i} \log \mathrm{NSW}_1(\A,\B)\right)_{i,\ell} 
    = \frac{\partial}{\partial B_j(i,\ell)} \log \mathrm{NSW}_1(\bm{A},\bm{B}) 
    = \frac{1}{U_i(\pi)} p_{ij}  \sum_{k=1}^{m} e_{k\ell} A_i(j,k).
\end{align*}

In the alternating maximization algorithm,
if $(\bm{A},\bm{B})$ converges to some value
$(\bm{A}^*,\bm{B}^*)$, then $\bm{A}^*$ and $\bm{B}^*$ are mutually NSW-best. From Theorem~\ref{thm:envy-free-policy}, we obtain a recommendation policy that satisfies double envy-freeness up to an additive difference of $\mathrm{O}(\varepsilon)$.
The experimental results presented in Section~\ref{sec:experiments} provide strong empirical evidence for the effectiveness of our approach. For detailed results, refer to Sections~\ref{sec:experiments:Synthetic_Data:results} and~\ref{sec:experiments:Real_World_Data:results}.

\section{Balancing the Trade-off between Fairness and Social Welfare}\label{sec:trade-off-fairness-socialwelfare}
The SW method maximizes the total number of matches, and the NSW method achieves fairness among agents. In practice, platform operators may want to balance these two objectives depending on their specific goals.
For example, a platform might accept a small amount of envy in exchange for significantly more matches.

To enable such flexibility, we propose the \emph{$\alpha$-social welfare ($\alpha$-SW)} function and $\alpha$-SW best recommendation, which interpolates between SW and NSW. This approach provides a flexible framework for balancing social welfare and fairness through a single tunable parameter $\alpha$.

For given $\alpha \in (0, 1]$, we define the $\alpha$-social welfare function for the left-side agents as $\mathrm{W}_1^{\alpha}(\A,\B) = \frac{1}{\alpha}\sum_{i=1}^n U_i(\A,\B)^{\alpha}$, and that for the right-side agents as $\mathrm{W}_2^{\alpha}(\A,\B) = \frac{1}{\alpha}\sum_{j=1}^m V_j(\A,\B)^{\alpha}$.
These $\alpha$-SW functions are continuous in $\alpha$ and correspond to social welfare ($\mathrm{SW}$) when $\alpha = 1$.
Intuitively, when $\alpha < 1$, the term $U_i(\A,\B)^\alpha$ grows sublinearly in utility, which means that agents with lower utilities contribute relatively more to the objective function. As a result, the optimization favors more balanced utility distributions. The smaller the $\alpha$, the stronger this equalizing effect becomes.

To analyze the limiting behavior as $\alpha \to 0$, we take the logarithm of $W_1^\alpha(\A, \B)$:
\begin{align*}
    \lim_{\alpha \to 0} \log W_1^\alpha(\A, \B) 
    &= 
        \lim_{\alpha \to 0} \frac{\log\left(\sum_{i=1}^n U_i(\A, \B)^\alpha\right)}{\alpha} 
    = 
        \lim_{\alpha \to 0} \frac{\sum_{i=1}^n U_i(\A, \B)^{\alpha}\log U_i(\A, \B)}{\sum_{i=1}^nU_i(\A, \B)^\alpha} 
    = 
        \frac{\sum_{i=1}^{n}\log U_i(\A, \B)}{n} \\
    &= 
        \frac{\log \mathrm{NSW}_1(\A, \B)}{n},
\end{align*}
where the second equality holds by L'H\^{o}pital's rule.
Since the logarithm is strictly increasing, maximizing $W_1^\alpha$ in the limit $\alpha \to 0$ is equivalent to maximizing $\mathrm{NSW}_1$. A similar argument applies to $W_2^\alpha$.

For given recommendations $\A$, a recommendation $\B^*$ is said to be \emph{$\alpha$-SW best} for the left side if $B^* \in \mathrm{argmax}_{\B} W_1^\alpha(\A, \B)$, and for given $B$, a recommendation $A^*$ is said to be \emph{$\alpha$-SW best} for the right side if $A^* \in \mathrm{argmax}_{\A} W_2^{\alpha}(\A, \B)$.

To compute mutually $\alpha$-SW best recommendations, we utilize alternating maximization through the Frank-Wolfe algorithm (Algorithm~\ref{alg:alternating_SW_maximization}) with $F_1(A, B) = W_1^\alpha(A, B)$ and $F_2(A, B) = W_2^\alpha(A, B)$. The gradient calculations are given by 
\begin{align*}
    \left( \nabla_{A_i}  \mathrm{W}_2^{\alpha}(\A,\B)\right)_{j,k} = \frac{\partial}{\partial A_i(j,k)}  \mathrm{W}_2^{\alpha}(\bm{A},\bm{B})
    = V_j(\A,\B)^{\alpha-1}  \frac{\partial V_j(\A,\B)}{\partial A_i(j,k)}  
    =  V_j(\A,\B)^{\alpha-1}  p_{ij}  \sum_{\ell=1}^{n} e_{k\ell} B_{j}(i,\ell),
\end{align*}
and 
\begin{align*}
    \left( \nabla_{B_j}  \mathrm{W}_1^{\alpha}(\A,\B)\right)_{i,\ell} = \frac{\partial}{\partial B_j(i,\ell)}  \mathrm{W}_1^{\alpha}(\A,\B)
    = U_i(\A,\B)^{\alpha-1}  \frac{\partial U_i(\A,\B)}{\partial B_j(i,\ell)}  
    = U_i(\A,\B)^{\alpha-1}  p_{ij}  \sum_{k=1}^{m} e_{k\ell} A_i(j,k).
\end{align*}

A mutually $\alpha$-SW best recommendation policy balances the trade-off between fairness and social welfare, and a larger $\alpha$ prioritizes social welfare while a smaller $\alpha$ prioritizes the fairness of recommendation opportunities.
The results of experiments for $\alpha$-SW methods are reported in Section~\ref{sec:experiments:Synthetic_Data_Experiment_II}.

\section{Faster Implementation using Sinkhorn Algorithm}\label{sec:Sinkhorn}
In Algorithm~\ref{alg:alternating_SW_maximization}, solving two linear programming (LP) subproblems over doubly stochastic matrices (lines~\ref{line:maximize_A_SW} and~\ref{line:maximize_B_SW}) becomes computationally prohibitive for large-scale problems due to the $O(nm^2 + n^2m)$ variable scaling. To address this challenge, we propose a computationally efficient approximation using the Sinkhorn algorithm~\cite{Sinkhorn1967}.

Each subproblem involves maximizing a linear objective over doubly stochastic matrices.
Since each doubly stochastic matrix is an independent variable, these optimizations can be performed separately. 
For instance, optimizing $\A$ requires solving $n$ independent problems:
$$
    \max_{X_i \in \mathbb{R}_{\geq 0}^{m \times m}} \sum_{j,k=1}^m \left(\nabla_{A_i} F_2(\mathbf{A},\mathbf{B})\right)_{j,k} \cdot \left( X_i\right)_{j,k} \quad \text{s.t.}\quad X_i \mathbf{1} = \mathbf{1}, \mathbf{1}^{\top} X_i = \mathbf{1}^{\top}.
$$ 
This is equivalent to an assignment problem which aims to minimize $\sum_{j,k=1}^m \left(C_i \right)_{j,k} \cdot \left( X_i\right)_{j,k}$ subject to $X$ being a doubly stochastic matrix, after transforming the cost matrix as $C_i =\max_{j,k\in [m]} \left(\nabla_{A_i} F_2(\bm{A},\bm{B}) \right)_{j,k} \cdot \boldsymbol{1} \boldsymbol{1}^{\top} - \nabla_{A_i} F_2(\bm{A},\bm{B})$, ensuring non-negative costs compatible with minimization.
Similarly, we solve $m$ independent assignment problems for optimizing $\mathbf{B}$.


\citet{Cuturi2013} showed that the Sinkhorn algorithm can be used to approximately solve the assignment problem with improved computational efficiency by introducing an entropic regularization term. 
Specifically, we introduce entropic regularization:
$$
    \min_{X \in \mathbb{R}_{\geq 0}^{m \times m}}  \left( \sum_{i,j=1}^m C_{ij} \cdot X_{ij} + \frac{1}{\tau} \sum_{i,j=1}^m X_{ij} \log X_{ij} \right) \quad \text{s.t.}\quad X_i \mathbf{1} = \mathbf{1}, \mathbf{1}^{\top} X_i = \mathbf{1}^{\top},
$$
where $\tau>0$ controls the trade-off between approximation fidelity and numerical stability. The larger $\tau$ weakens the regularization, approaching the original LP, but risks instability.

The dual problem of this entropy-regularized optimization can be solved using an iterative update procedure, known as the Sinkhorn algorithm, which we present in Algorithm~\ref{alg:Sinkhorn}.
The algorithm first elementwisely exponentiates the cost matrix to obtain $K = \exp(-\tau C)$, and iteratively applies the normalization of rows and columns (lines~\ref{line:alg:Sinkhorn:for_start}-\ref{line:alg:Sinkhorn:for_end}) until convergence to an approximate doubly stochastic matrix $X = \operatorname{diag}(u) \, K \, \operatorname{diag}(v)$. 
This matrix scaling approach circumvents expensive LP solvers, achieving a per-iteration complexity of $O(m^2)$ with highly parallelizable matrix operations suitable for GPU acceleration.

\begin{algorithm}[t]
    \caption{\textsc{Sinkhorn Algorithm}~\cite{Cuturi2013}}
    \label{alg:Sinkhorn}
    \begin{algorithmic}[1]
        \REQUIRE A cost matrix $C \in \mathbb{R}_{\ge 0}^{m \times m}$, and a hyper-parameter $\tau > 0$.
        \ENSURE An approximate doubly stochastic matrix $X \in \mathbb{R}_{\ge 0}^{m \times m}$.
        \STATE \(K \gets \exp(-\tau C)\)  (elementwise exponentiation) 
        \STATE Initialize vectors \(u \gets \mathbf{1} \in \mathbb{R}^m\) and \(v \gets \mathbf{1} \in \mathbb{R}^m\).
        \FOR{$t=1,2,\ldots,T$}\label{line:alg:Sinkhorn:for_start}
            \STATE 
                \(u \gets \mathbf{1} \oslash (K v)\), where \(\oslash\) denotes elementwise division.
            \STATE 
                \(v \gets \mathbf{1} \oslash (K^\top u)\)
        \ENDFOR\label{line:alg:Sinkhorn:for_end}
        \STATE $X \gets \operatorname{diag}(u) \, K \, \operatorname{diag}(v)$
        \RETURN $X$
    \end{algorithmic}
\end{algorithm}

Algorithm~\ref{alg:alternating_Sinkhorn} integrates the Sinkhorn algorithm into the alternating maximization framework of Algorithm~\ref{alg:alternating_SW_maximization}.
Crucially, the $n$ subproblems for $\A$ and $m$ for $\B$ can be batched, allowing for parallel computation. 
Lines \ref{line:alg:alternating_Sinkhorn:for_start_A}-\ref{line:alg:alternating_Sinkhorn:for_end_A} and \ref{line:alg:alternating_Sinkhorn:for_start_B}-\ref{line:alg:alternating_Sinkhorn:for_end_B} replace LP solvers with Sinkhorn steps.
Although the use of entropic regularization introduces some approximation error, our empirical results (Section~\ref{sec:experiments:Sinkhorn}) demonstrate that the approach produces competitive performance. Specifically, it scales effectively to larger problem sizes---where traditional LP solvers struggle---while achieving lower envy and higher social welfare. More details are provided in Section~\ref{sec:experiments:Sinkhorn}.

\begin{algorithm}[t]
    \caption{Alternating Maximization via Frank-Wolfe Algorithm using the Sinkhorn Algorithm}
    \label{alg:alternating_Sinkhorn}
    \begin{algorithmic}[1]
        \REQUIRE Preference probabilities $(p_{ij})_{i,j}$, an examination function $e(\cdot)$, and learning rates $(\eta_t)_{t\in [T]}$.
        \ENSURE Recommendations $\bm{A},\bm{B}$.
        \STATE Initialize $\bm{A},\bm{B}$.
        \FOR{$t=1,2,\ldots,T$}
            \FOR{$i=1,2,\ldots,n$} \label{line:alg:alternating_Sinkhorn:for_start_A}
                \STATE 
                    $C_i \gets \max_{j,k\in [m]} \left(\nabla_{A_i} F_2(\bm{A},\bm{B}) \right)_{j,k} \cdot \boldsymbol{1} \boldsymbol{1}^{\top} - \nabla_{A_i} F_2(\bm{A},\bm{B})$
                \STATE 
                    $X^*_i \gets \textsc{Sinkhorn Algorithm}(C_i)$
            \ENDFOR \label{line:alg:alternating_Sinkhorn:for_end_A}
            \STATE 
                $\bm{A} \leftarrow (1-\eta_t)\bm{A}  + \eta_t \bm{X}^*$, where $\bm{X}^* = (X^*_1,\ldots, X^*_n)$.
            \FOR{$j=1,2,\ldots,m$} \label{line:alg:alternating_Sinkhorn:for_start_B}
                \STATE 
                    $C_j \gets \max_{i,\ell\in [n]} \left(\nabla_{B_j} F_1(\bm{A},\bm{B}) \right)_{j,k} \cdot \boldsymbol{1} \boldsymbol{1}^{\top} - \nabla_{B_j} F_1(\bm{A},\bm{B})$
                \STATE 
                    $Y^*_j \gets \textsc{Sinkhorn Algorithm}(C_j)$
            \ENDFOR \label{line:alg:alternating_Sinkhorn:for_end_B}
            \STATE 
                $\bm{B} \leftarrow (1-\eta_t)\bm{B}  + \eta_t \bm{Y}^*$, where $\bm{Y}^* = (Y^*_1,\ldots, Y^*_m)$.
        \ENDFOR
        \RETURN $\bm{A},\bm{B}$.
    \end{algorithmic}
\end{algorithm}

\section{Experiments}\label{sec:experiments}
In this section, we evaluate our proposed methods through experiments using synthetic data and real-world data from a Japanese online dating platform.

\subsection{Preliminaries}\label{sec:experiments:Preliminaries}

Before conducting experiments, we introduce the baseline methods and provide the Gini index used in our evaluation.

\subsubsection{Baseline Methods}\label{sec:experiments:Baseline_Methods}

In our experiments, we compare our proposed approaches with several baseline methods: Naive, Prod, TU, and IterLP. 
The Naive and Prod methods represent conventional approaches widely used in reciprocal recommender systems~\cite{neve2019latent,neve2019aggregation}.
The TU matching method was proposed by \citet{Tomita2023} as an approach to reciprocal recommendation in matching markets where the roles of the two sides differ; this method can also be adapted to our model.
Finally, IterLP is a greedy approach that we propose inspired by the similarity between our model and maximum matching problems.
The details of each method are described as follows.

\paragraph{Naive} We recommend right-side agents to each left-side agent $a_i$ in the non-increasing order of the estimated preferences $\hat{p}_1(i, :)$ deterministically.
The same procedure is applied to recommend agents on the left side to each agent on the right side.

\paragraph{Prod} We first compute the reciprocal score for each pair $(i, j)$ as $p_{ij} = \hat{p}_1(i, j) \cdot \hat{p}_2(j, i)$, and then recommend right-side agents to each left-side agent $a_i$ in non-increasing order of these scores.
The same procedure is used to recommend left-side agents to each right-side agent.

\paragraph{TU (A Heuristic Method Based on TU Matching)}
This method is based on the transferable utility (TU) model, which has been extensively studied in economics~\cite{shapley1971assignment,choo2006who}. 
Recently, \citet{Tomita2023} proposed a deterministic recommendation policy using this framework to increase the expected number of matches while reducing popularity bias in reciprocal recommender systems.
Although their work focused on reciprocal recommendations with slightly different assumptions about the roles of the two sides from ours, their underlying TU method can be adapted to our setting. 
Algorithm~\ref{alg:TU} details this approach, which computes a matching score $\mu_{ij}$ for each $(i,j)$ pair.
These scores are then used to generate recommendations: for each left-side agent $a_i$, right-side agents are recommended in descending order of $\mu_{ij}$, with an analogous process for right-side agents.

\begin{algorithm}[t]
    \caption{A Heuristic Method Based on TU Matching (TU)~\cite{Tomita2023}}
    \label{alg:TU}
    \begin{algorithmic}[1]
        \renewcommand{\algorithmicrequire}{\textbf{Input:}}
        \renewcommand{\algorithmicensure}{\textbf{Output:}}
        \REQUIRE Preference probabilities $(\hat{p}_1(i,j))_{i,j},~(\hat{p}_2(j,i))_{j,i}$, a scale parameter $\beta > 0$, and timesteps $T$.
        \ENSURE Matching scores $(\mu_{ij})_{i, j}$
        \STATE Initialize $X_i= 1,Y_j=1$ for all $i\in N,j\in M$.
        \FOR{$t = 1, \dots, T$}
        \STATE $X_i \leftarrow \sqrt{1 + \left(\frac{1}{2} \sum_{j = 1}^M \exp\left(\frac{\hat{p}_1(i,j)+\hat{p}_2(j, i)}{2\beta}\right)Y_j\right)^2} - \frac{1}{2} \sum_{j \in M} \exp\left(\frac{\hat{p}_1(i,j)+\hat{p}_2(j, i)}{2\beta}\right)Y_j$ for each $i \in N$
        \STATE $Y_j \leftarrow \sqrt{1 + \left(\frac{1}{2} \sum_{i\in N} \exp\left(\frac{\hat{p}_1(i,j)+\hat{p}_2(j, i)}{2\beta}\right)X_i\right)^2} - \frac{1}{2} \sum_{i \in N} \exp\left(\frac{\hat{p}_1(i,j)+\hat{p}_2(j, i)}{2\beta}\right)X_i$ for each $j \in M$
        \ENDFOR
        \STATE Let $\mu_{ij} = \exp\left(\frac{\hat{p}_1(i,j)+\hat{p}_2(j, i)}{2\beta}\right)X_i Y_j$ for each $(i, j) \in N \times M$.
        \RETURN $\mu$.
    \end{algorithmic}
\end{algorithm}

\paragraph{IterLP (A Heuristic Method Based on the Matching Problem)}
We note the connection between the matchings on a bipartite graph and the reciprocal recommendations for both sides and propose a method based on the maximum weight matchings.
In fact, a matching on the complete bipartite graph between $N$ and $M$ with edge-weights $(p_{ij})_{i,j}$ provides recommendations. 
Here, for some $k \in [K]$, $A_{i}(j,k) = 1$ and $B_{j}(i,k) = 1$ if and only if the edge $\{i,j\}$ is included in the matching.
It is worth noting that, in the special case where $K = 1$, $n = m$ and only deterministic policies are considered, a policy induced from a maximum weight matching is socially optimal among deterministic policies. Furthermore, the policy satisfies double envy-freeness since if agent $a_i\in N$ matches with agent $b_j\in M$, then we have $B_j(i',1)=0$ for $i'\neq i$ and another agent $a_{i'} \in N$ never envies agent $a_{i}$. A similar discussion holds for agents in $M$. 

We propose an algorithm based on maximum weight matchings (Algorithm~\ref{alg:IterLP}) that iteratively solves a linear programming (LP) relaxation problem for maximum weight matching, filling up each position $k$ of recommendations. At each iteration, the algorithm recommends the partner most preferred for each agent based on the solution of the LP. 


\begin{algorithm}[t]
    \caption{Iterating Linear Programming of Maximum Weight Matching (IterLP)}
    \label{alg:IterLP}
    \begin{algorithmic}[1]
        \renewcommand{\algorithmicrequire}{\textbf{Input:}}
        \renewcommand{\algorithmicensure}{\textbf{Output:}}
        \REQUIRE Preference probabilities $(p_{ij})_{i,j}$.
        \ENSURE Recommendations $\bm{A},\bm{B}$.
        \STATE Initialize $A_i=O,B_j=O$ for all $i\in N,j\in M$.
        \FOR{$k=1,2,\ldots,K$} 
        \STATE Let $(x_{i,j}^*)_{i,j}$ be an LP solution of the maximum weight matching between $N$ and $M$ with edge-weights $(p_{ij})_{i,j}$.
        \STATE For each $i\in N$, let $j^*_i = \max_{j\in M} x_{i,j}^*$, and set $A_i(j^*_i,k)=1$.
        \STATE For each $j\in M$, let $i^*_j = \max_{i\in N} x_{i,j}^*$, and set $B_j(i^*_j ,k)=1$.
        \STATE Set the weights of $\{i,j^*_i\}$ and $\{i^*_j,j\}$ to $-\infty$.
        \ENDFOR \label{line:for_end}
        \RETURN $\bm{A},\bm{B}$.
    \end{algorithmic}
\end{algorithm}

\subsubsection{The Gini Index}\label{sec:experiments:Gini_index}

While we focus mainly on the fairness of recommended opportunities among users, we also consider the fairness in terms of distribution of resulting utilities (the expected number of matches that one user gets).
One of the important measures to consider the fairness of utility distributions is the Gini index~\cite{gini1936measure}.
The Gini index is commonly used to measure inequalities in income distributions in economics~\cite{sen1997economic}.
Given a recommendation policy $(\A, \B)$, we calculate the utility of each left-side agent $u_i = U_i(\A, \B)$.
The Gini index for left-side agents is computed by 
\[
    G_1 = \frac{\sum_{i=1}^n\sum_{j=1}^n |u_i - u_j|}{2n\sum_{i=1}^n u_i}.
\]
The Gini index is twice as large as an area between the $45$ degree line and the Lorenz curve of utilities.
The larger the Gini index becomes, the more unequal the distribution of utilities.
$G_1 = 0$ means perfect equality (where all agents on the left side get the same level of utility), while $G_1 = 1$ means perfect inequality (where only one left-side agent gets a positive utility).
The Gini index for the right side $G_2$ can be computed similarly.

\subsection{Synthetic Data Experiment I: Comparison with Different Popularity Levels}
\label{sec:experiments:Synthetic_Data_Experiment_I}

In this section, we evaluate our methods on a synthetic data experiment that simulates an online dating platform.

\subsubsection{Data generation}
We first explain the procedure used to generate the synthetic data.
We set the number of agents on the right side as $m = 50$, and consider two cases for the number of agents on the left side $n$: the balanced case with $n = 50$ and the unbalanced case with $n = 75$.
To control the level of popularity differences among agents, we generate preference probabilities with two terms and a parameter $\lambda$, following \citet{su2022optimizing} and~\citet{Tomita2023}:
$$
    \hat{p}_1(i,j) = \lambda \cdot \hat{p}_{1}^{\mathrm{pop}}(j) + (1-\lambda)\cdot \hat{p}_{1}^{\mathrm{unif}}(i,j)
$$
and 
$$
    \hat{p}_2(j,i) = \lambda \cdot \hat{p}_{2}^{\mathrm{pop}}(i) + (1-\lambda)\cdot \hat{p}_{2}^{\mathrm{unif}}(j,i),
$$
where $\hat{p}_{1}^{\mathrm{pop}}(j) = \frac{j-1}{m-1}$ and $\hat{p}_{2}^{\mathrm{pop}}(i) = \frac{i-1}{n-1}$ represent the overall popularity of  agent $b_j$ and agent $a_i$, respectively. 
In contrast, let $\hat{p}_{1}^{\mathrm{unif}}(i,j)$ and $\hat{p}_{2}^{\mathrm{unif}}(j,i)$ represent individual preferences drawn independently from the uniform distribution on $[0,1]$.
The parameter $\lambda$ controls the level of differences in popularity, and we consider $\lambda \in \{0.0, 0.2, 0.4, 0.6, 0.8, 1.0\}$. 
When $\lambda$ is large, agents on each side tend to have very similar preferences towards agents on the opposite side, leading to significant popularity differences.
Finally, we test two examination functions: $e(k) = \frac{1}{\log_2(k+1)}$ (denoted as ``log'') or $e(k) = \frac{1}{k}$ (denoted as ``inv'').

\subsubsection{Procedure}
For each case of $n$, $\lambda$, and the examination function, we generate 10 sample sets of preference probabilities, where random seeds are set to $0,1,\dots,9$.\footnote{For details of preference generation process, see the source code \url{https://github.com/CyberAgentAILab/FairReciprocalRecommendation/blob/main/src/experiments.py}.}
For each sample, we compare the following recommendation methods:
\begin{itemize}
    \item \textbf{SW}: Alternating social welfare maximization (Algorithm~\ref{alg:alternating_SW_maximization} with $F_1(\bm{A},\bm{B})=F_2(\bm{A},\bm{B})=\mathrm{SW}(\bm{A},\bm{B})$),
    \item \textbf{NSW}: Alternating Nash social welfare maximization (Algorithm~\ref{alg:alternating_SW_maximization} with $F_1(\bm{A},\bm{B})=\log \mathrm{NSW}_1(\bm{A},\bm{B})$ and $F_2(\bm{A},\bm{B})=\log \mathrm{NSW}_2(\bm{A},\bm{B})$), and
    \item \textbf{Baseline methods}: Naive, Prod, IterLP, and TU (as presented in Section~\ref{sec:experiments:Baseline_Methods}).
\end{itemize}
In the SW, NSW, and IterLP methods, we used CVXPY~\cite{diamond2016cvxpy} with the ECOS solver to solve the inner linear programming problems.
For the SW and NSW methods, we set the learning rate $\eta = 0.1$ and the number of iterations $T = 100$, ensuring that the updating variables $A$ and $B$ converge before termination in all cases\footnote{We determined that the updates are converged when the maximum update difference of expected matches is less than 0.01.}.
For each sample, we computed (i) the expected number of matches, (ii) the number of pairs $(i, i')$ of left-side agents for which $a_i$ envies $a_{i'}$, and (iii) the number of pairs $(j, j')$ of right-side agents for which $b_j$ envies $b_{j'}$.

\subsubsection{Results}\label{sec:experiments:Synthetic_Data:results}

\begin{figure}[t]
    \begin{minipage}[t]{\linewidth}
        \centering
        \includegraphics[width=\linewidth]{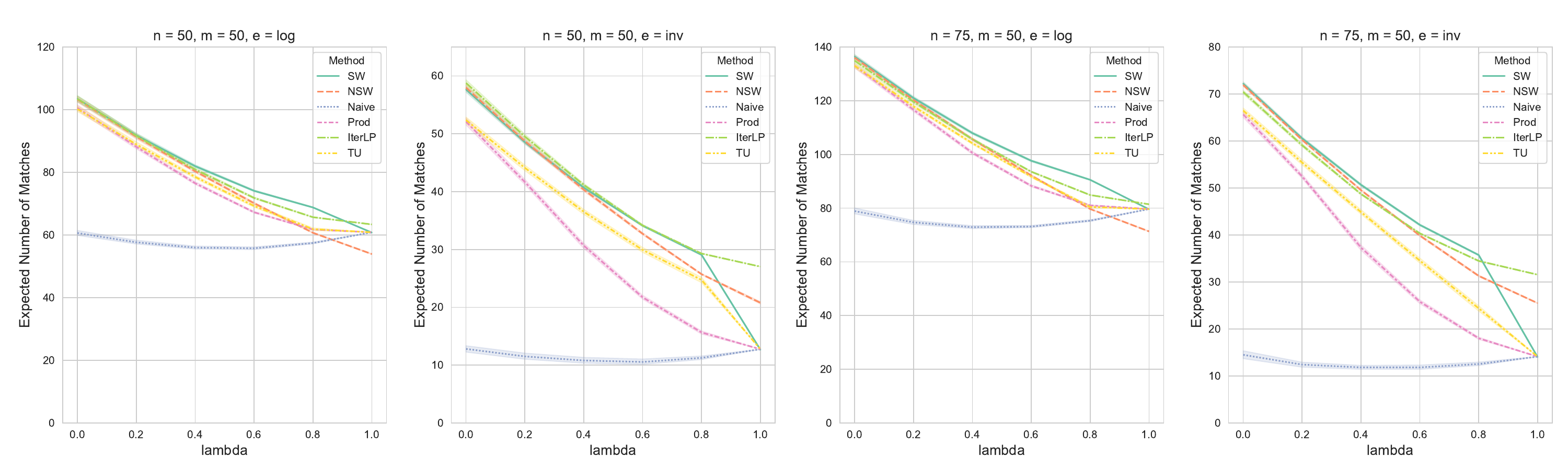}
        \subcaption{Expected number of matches.}
    \end{minipage}
    \begin{minipage}[t]{\linewidth}
        \centering
        \includegraphics[width=\linewidth]{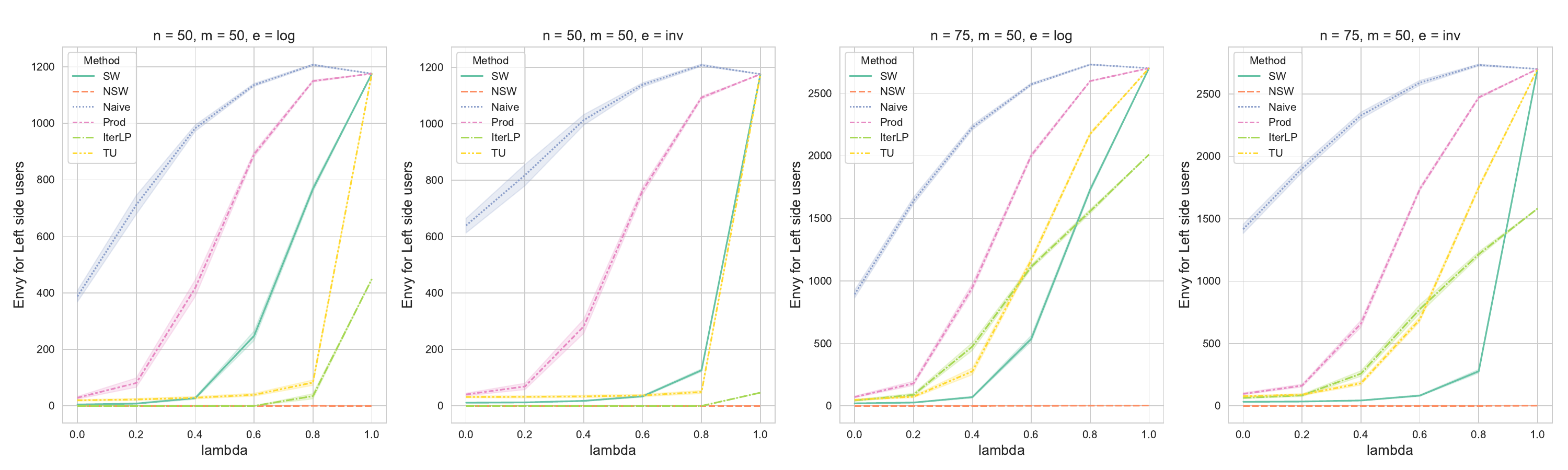}
        \subcaption{Envy for left-side agents.}
    \end{minipage}
    \begin{minipage}[t]{\linewidth}
        \centering
        \includegraphics[width=\linewidth]{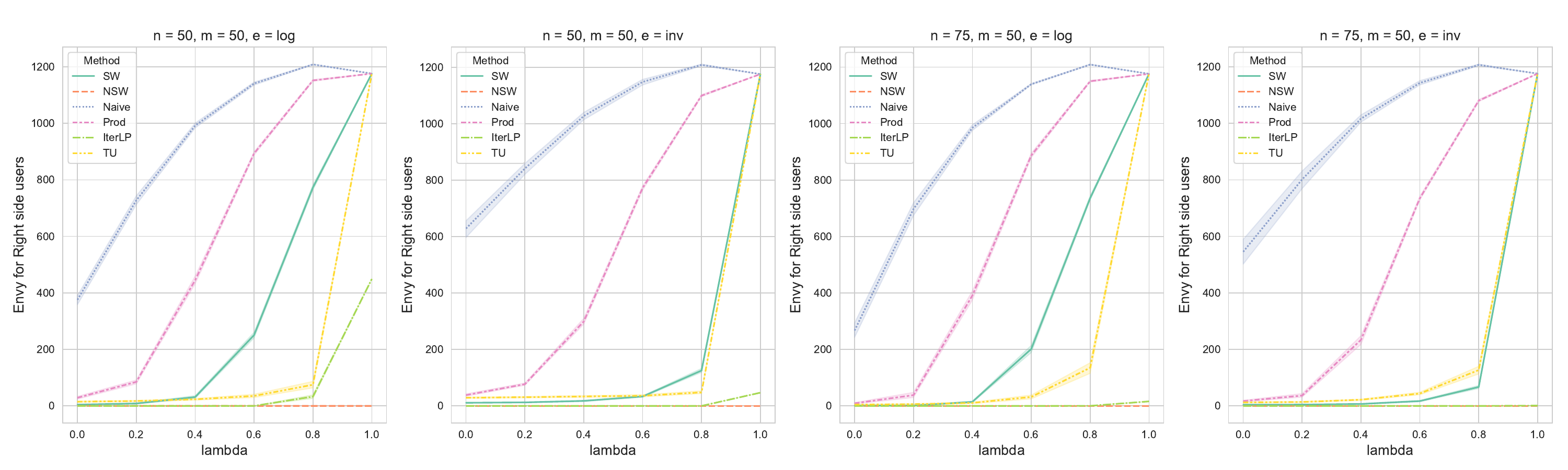}
        \subcaption{Envy for right-side agents.}
    \end{minipage}
    \caption{Results of experiments on synthetic data.
    Upper, middle and lower rows show the expected number of matches, the number of envies for left-side agents, and that for right-side agents for each case, respectively.
    We vary the number of left-side agents $n \in \{50, 75\}$, the examination functions $e(k) = 1/\log_2(k+1)$ (log) or $1/k$ (inv), and lambda $\lambda \in \{0.0, 0.2, \dots, 1.0\}$, while we fix the number of right-side agents $m = 50$, where each case correspond to each column.
    We conducted 10 experiments for each case, and report the mean of results and its 95\% confidence intervals (but it is invisible in many cases due to its small variations.).}
    \label{fig:synthetic_data}
\end{figure}

\begin{table}[t]
\caption{
Comparison of different methods with varying $\lambda$ on synthetic datasets with $n=75$, $m=50$, and $e = \log$. 
Results are averaged over 10 independent trials. 
The table reports the mean and standard deviation of three metrics: the expected number of matches (Expected Matches), the number of envious pairs among left-side agents (Left Envy), and among right-side agents (Right Envy).
As shown in the table, NSW maintains low envy levels among left-side agents  even in high popularity-difference scenarios, with envy measures of $0.70 \pm 0.82$, $2.10 \pm 1.37$, and $2.50 \pm 0.52$ for $\lambda = 0.6$, $0.8$, and $1.0$, respectively.
We conducted the Wilcoxon signed-rank test comparing Left/Right Envy of the NSW method with those of SW ($^*p < 0.05,~ ^{**}p < 0.01,~ ^{***}p < 0.001$).
}
\label{tb:syenthetic-res-75-50-log}
\centering
\begin{tabular}{llrrr}
\toprule
 $\lambda$ & Method & Expected Matches & Left Envy &  Right Envy \\
\midrule
\multirow[t]{6}{*}{0.0} 
 & Naive & 78.9\plmi 1.75 & 892.9\plmi 54.0 & 267.2\plmi 41.4 \\
 & Prod & 132.9\plmi 1.33 & 71.2\plmi 11.6 & 9.30\plmi 4.21 \\
 & TU & 133.0\plmi 1.32 & 48.6\plmi 7.76 & 4.50\plmi 1.90 \\
 & IterLP & 134.9\plmi 1.46 &  43.4\plmi 10.8 &  0.00\plmi 0.00 \\
 & SW & 136.4\plmi 1.27 & 20.5\plmi 3.59 & 0.90\plmi 0.99 \\
 & NSW & 136.0\plmi 1.40 & 0.10\plmi 0.31*** & 0.00\plmi 0.00* \\
\cline{1-5}
\multirow[t]{6}{*}{0.2} 
 & Naive & 74.70\plmi 1.24 & 1639.1\plmi 48.7 & 697.4\plmi 40.3 \\
 & Prod & 116.6\plmi 1.07 & 179.7\plmi 28.4 & 38.2\plmi 16.1 \\
 & TU & 117.8\plmi 0.93 & 73.8\plmi 13.8 & 6.20\plmi 2.74 \\
 & IterLP & 119.6\plmi 1.07 & 88.4\plmi 19.3 & 0.00\plmi 0.00 \\
 & SW & 121.0\plmi 1.05 & 26.8\plmi 3.61 & 1.00\plmi 0.94 \\
 & NSW & 120.2\plmi 1.05 & 0.10\plmi 0.31*** & 0.00\plmi 0.00** \\
\cline{1-5}
\multirow[t]{6}{*}{0.4} 
 & Naive & 72.9\plmi 0.91 & 2227.4\plmi 35.7 & 985.0\plmi 16.9 \\
 & Prod & 100.5\plmi 0.86 & 947.9\plmi 48.0 & 391.0\plmi 27.3 \\
 & TU & 104.1\plmi 0.69 & 276.4\plmi 45.0 & 10.7\plmi 4.21 \\
 & IterLP & 105.6\plmi 0.70 & 474.9\plmi 55.2 & 0.00\plmi 0.00 \\
 & SW & 107.9\plmi 0.79 & 70.0\plmi 11.7 & 14.4\plmi 4.24 \\
 & NSW & 105.6\plmi 0.71 & 0.20\plmi 0.42*** & 0.00\plmi 0.00*** \\
\cline{1-5}
\multirow[t]{6}{*}{0.6}
 & Naive & 73.1\plmi 0.67 & 2571.3\plmi 17.1 & 1139.3\plmi 3.59 \\
 & Prod & 88.2\plmi 0.66 & 2002.4\plmi 25.7 & 887.0\plmi 13.3 \\
 & TU & 91.7\plmi 0.62 & 1164.4\plmi 33.0 & 31.7\plmi 11.6 \\
 & IterLP & 93.5\plmi 0.38 & 1113.1\plmi 29.1 & 0.10\plmi 0.31 \\
 & SW & 97.6\plmi 0.57 & 536.0\plmi 43.0 & 201.6\plmi 21.5 \\
 & NSW & 92.1\plmi 0.49 & \textbf{0.70\plmi 0.82}*** & 0.00\plmi 0.00*** \\
\cline{1-5}
\multirow[t]{6}{*}{0.8} 
 & Naive & 75.3\plmi 0.50 & 2730.0\plmi 6.73 & 1208.9\plmi 2.92 \\
 & Prod & 81.0\plmi 0.35 & 2597.9\plmi 5.56 & 1149.9\plmi 4.55 \\
 & TU & 80.4\plmi 0.72 & 2175.7\plmi 19.6 & 135.1\plmi 32.8 \\
 & IterLP & 84.8\plmi 0.24 & 1556.9\plmi 27.5 & 0.20\plmi 0.63 \\
 & SW & 90.5\plmi 0.30 & 1728.9\plmi 29.1 & 735.7\plmi 11.7 \\
 & NSW & 79.6\plmi 0.24 & \textbf{2.10\plmi 1.37}*** & 0.00\plmi 0.00*** \\
\cline{1-5}
\multirow[t]{6}{*}{1.0} 
 & Naive & 79.6\plmi 0.00 & 2701.0\plmi 0.00 & 1176.0\plmi 0.00 \\
 & Prod & 79.6\plmi 0.00 & 2701.0\plmi 0.00 & 1176.0\plmi 0.00 \\
 & TU & 79.6\plmi 0.00 & 2701.0\plmi 0.00 & 1176.0\plmi 0.00 \\
 & IterLP & 81.4\plmi 0.00 & 2010.0\plmi 0.00 & 16.0\plmi 0.00 \\
 & SW & 79.5\plmi 0.00 & 2701.0\plmi 0.00 & 1176.0\plmi 0.00 \\
 & NSW & 71.3\plmi 0.20 & \textbf{2.50\plmi 0.52}*** & 0.00\plmi 0.00*** \\
\bottomrule
\end{tabular}
\end{table}

The results of experiments with synthetic data are summarized in Table~\ref{tb:syenthetic-res-75-50-log} and Fig.~\ref{fig:synthetic_data}.

First, we note that as the popularity difference level $\lambda$ increases, the expected matches decrease and the number of envies increases for most methods.
As the popularity differences become more pronounced, the recommendation problem becomes increasingly difficult due to user congestion around popular agents.
The SW method achieves the highest expected matches in almost all cases except for $\lambda = 1.0$, while it causes relatively many envies when $\lambda \ge 0.6$.
Thus, we empirically showed that there exists a trade-off between the social welfare and the fairness of recommendation opportunities when there are popularity differences among users.

Second, the NSW method demonstrates remarkable performance in terms of both fairness and efficiency.
It maintains near-zero envy across all scenarios while achieving competitive expected match rates compared to SW and other methods.
Most notably, even with high popularity differences ($\lambda = 0.8$ and $\lambda=1.0$), where other methods exhibit significant unfairness, NSW continues to maintain high social welfare and near-zero envy levels.
We conducted the Wilcoxon signed-rank test comparing the number of envies for left and right users of the NSW method with those of SW method in Table~\ref{tb:syenthetic-res-75-50-log}, and verified statistical significance in all cases at $p < 0.05$.

Finally, our empirical analysis shows that the heuristic approach IterLP (Algorithm~\ref{alg:IterLP}) performs well in balanced scenarios, achieving relatively high social welfare while maintaining low levels of envy.
However, its performance deteriorates in the unbalanced case ($n = 75$, $m = 50$), particularly failing to mitigate envy among left-side agents.
In contrast, NSW demonstrates robust performance, maintaining near-zero envy for agents on both sides, as detailed in Table~\ref{tb:syenthetic-res-75-50-log}.

\begin{figure}[t]
    \begin{tabular}{cc}
    \begin{minipage}[t]{0.48\textwidth}
    \centering
    \includegraphics[width=1.0\linewidth]{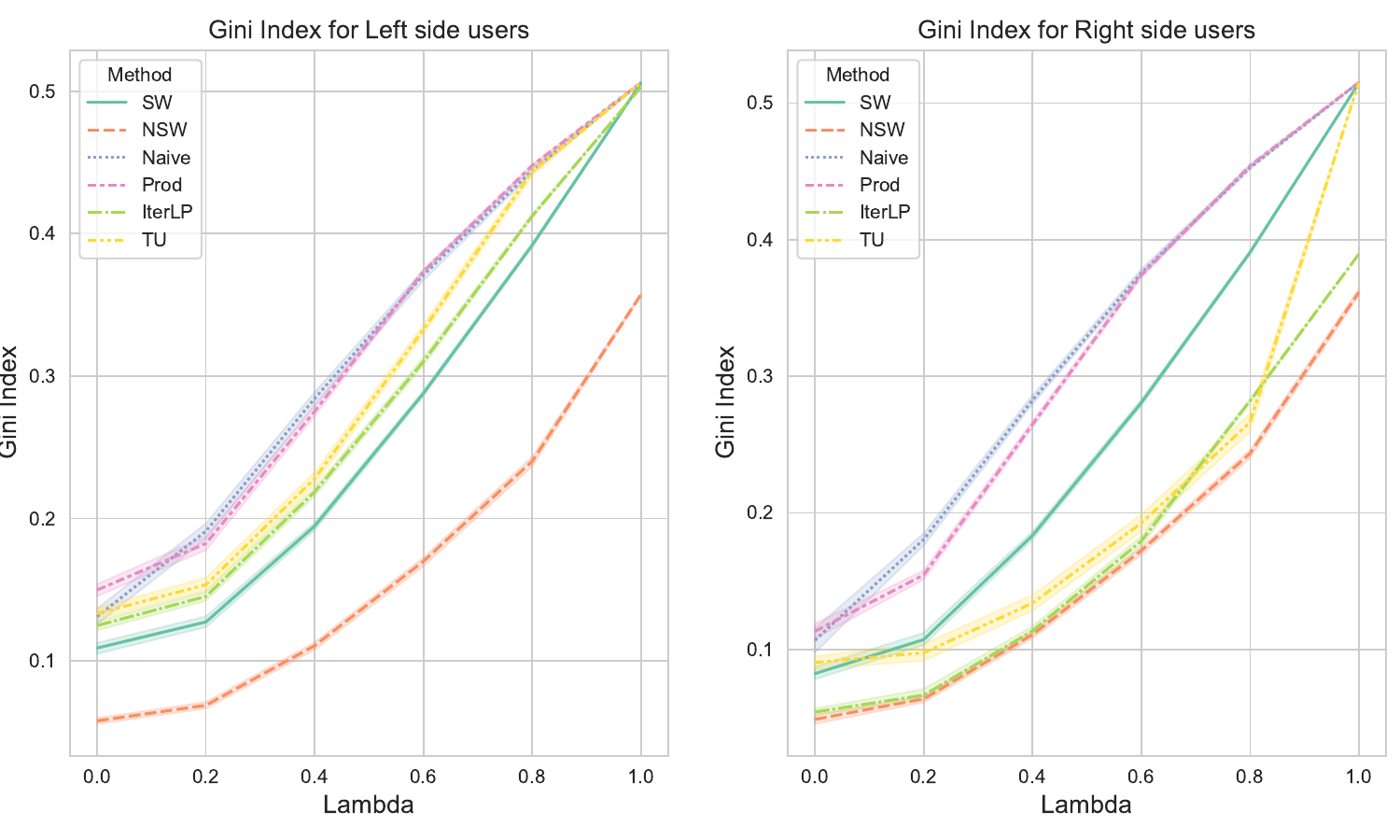}
    \subcaption{Gini index of expected matches for e = `log'.}
    \label{fig:synthetic-gini-log}
    \end{minipage}
    \begin{minipage}[t]{0.48\textwidth}
    \centering
    \includegraphics[width=1.0\linewidth]{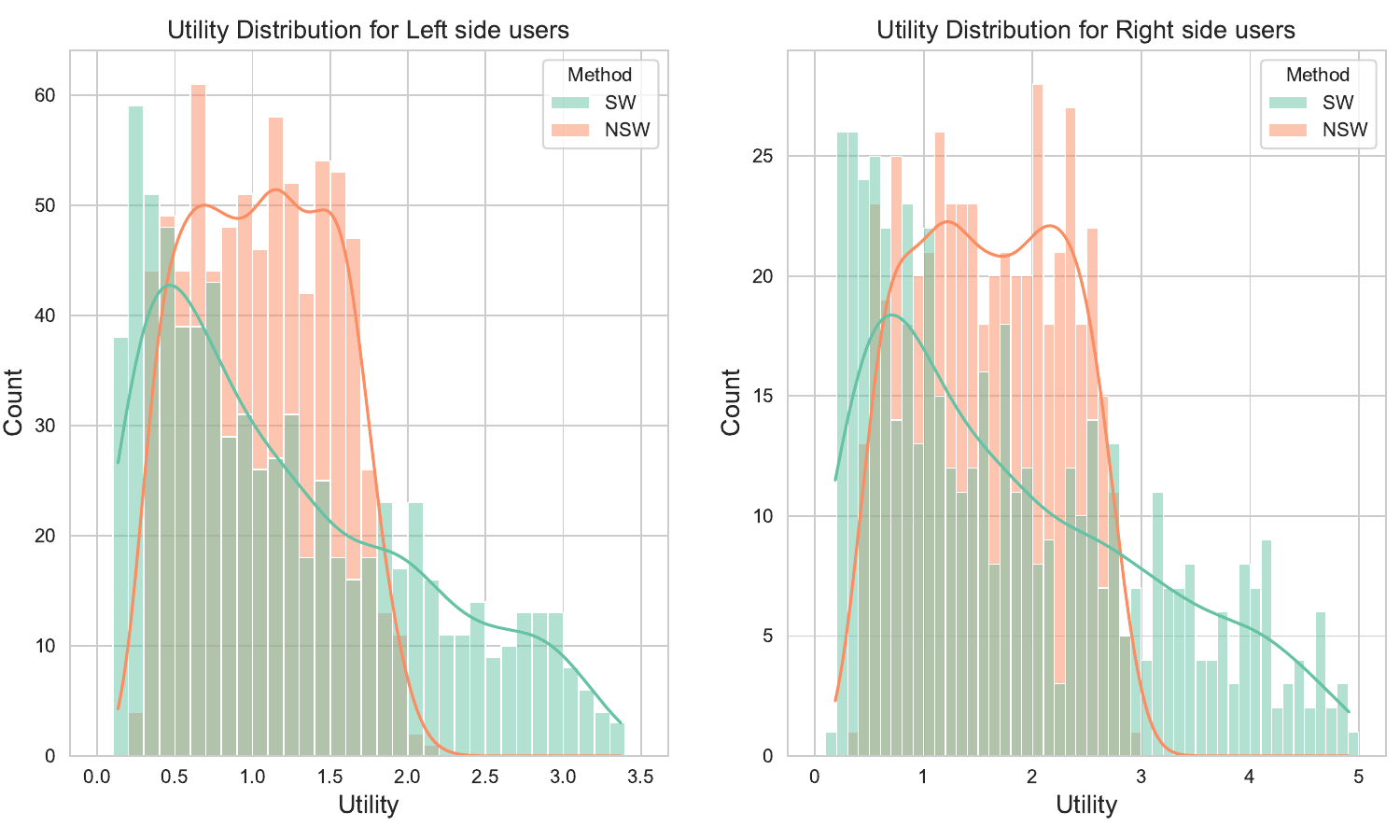}
    \subcaption{Utility distributions of 10 experiments for $\lambda = 0.8$, e = `log'.}
    \label{fig:synthetic-dist-log}
    \end{minipage}
    \\
    \begin{minipage}[t]{0.48\textwidth}
    \centering
    \includegraphics[width=1.0\linewidth]{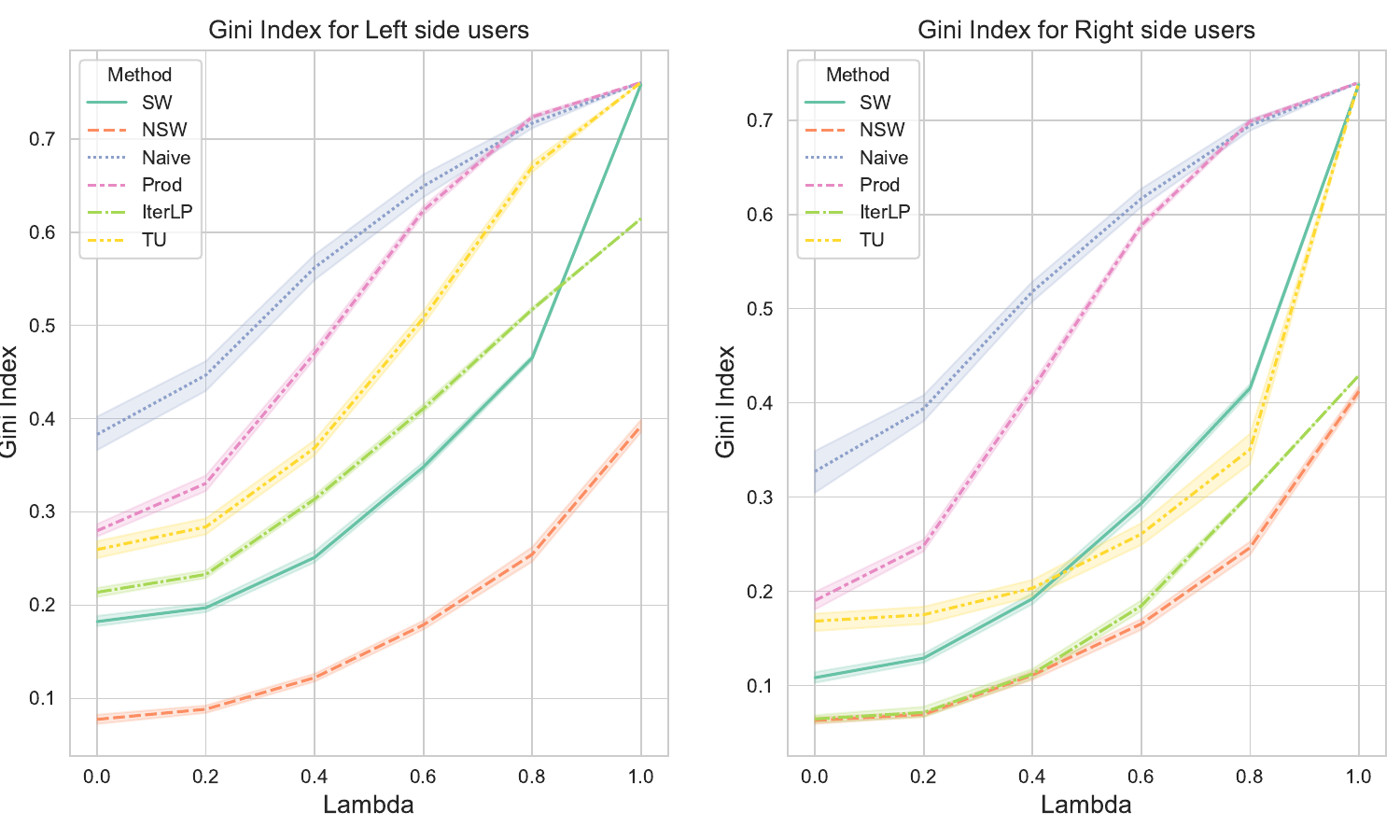}
    \subcaption{Gini index of expected matches for e = `inv'.}
    \label{fig:synthetic-gini-inv}
    \end{minipage}
    \begin{minipage}[t]{0.48\textwidth}
    \centering
    \includegraphics[width=1.0\linewidth]{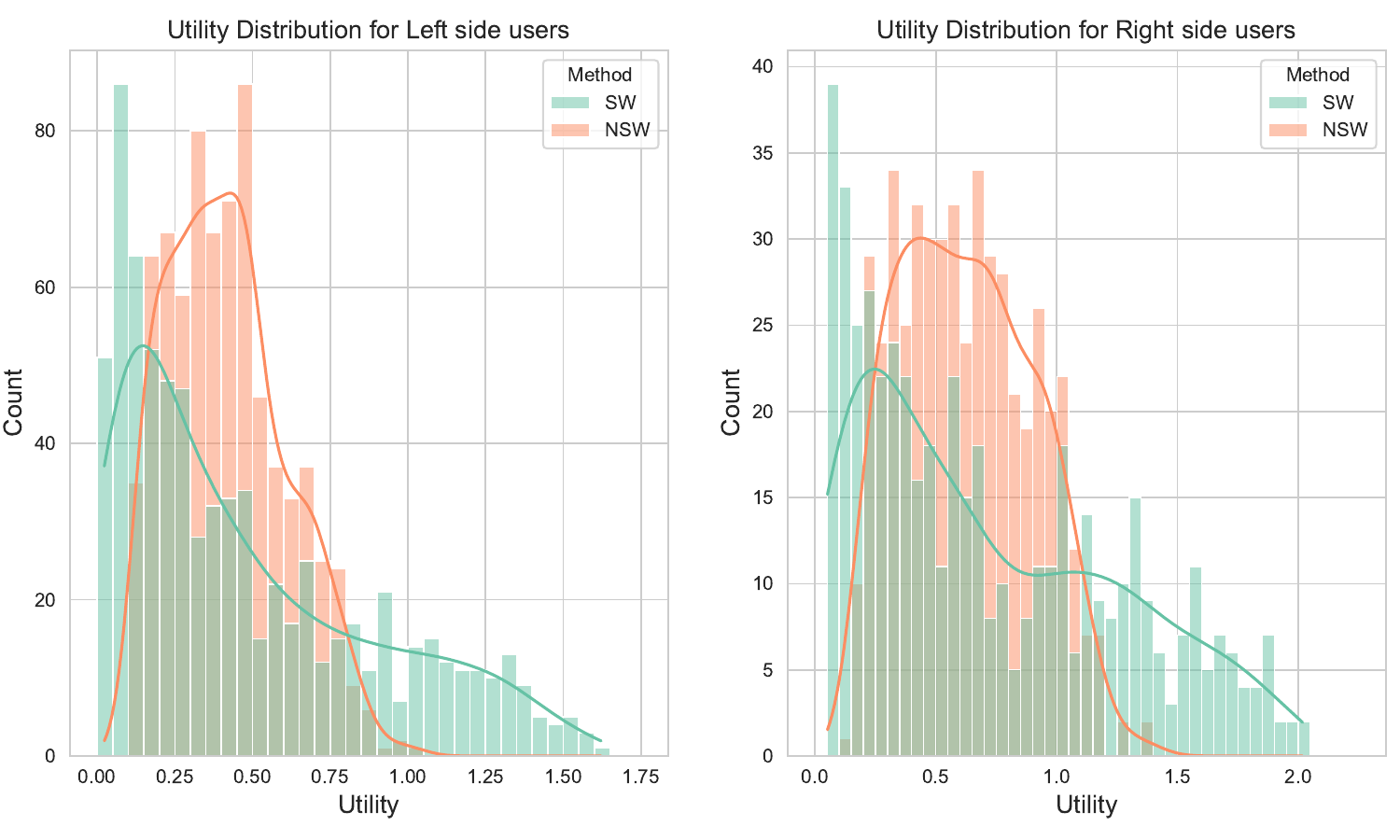}
    \subcaption{Utility distributions of 10 experiments for $\lambda = 0.8$, e = `inv'.}
    \label{fig:synthetic-dist-inv}
    \end{minipage}
    \end{tabular}
    \caption{The four left figures show the Gini index of users' utilities (expected number of matches) in the synthetic data experiment for the case where $n = 75$, $m = 50$, $\lambda \in \{0.0,0.2,0.4,0.6,0.8,1.0\}$, and $e = \text{``log''}$ (a) or $e =  \text{``inv''}$ (c). Mean value over 10 trials and 95\% CI are reported. 
    Four right figures are the utility distributions over all of 10 trials in the case of $n = 75$, $m = 50$, $\lambda = 0.8$, and $e = \text{``log''}$ (b) or $e =  \text{``inv''}$ (d).
    }
    \label{fig:other-fairness-criteria}
\end{figure}

In Fig.~\ref{fig:synthetic-gini-log} and \ref{fig:synthetic-gini-inv}, we report the Gini index of the user utilities for each method in the case of $n = 75, m = 50$ and the examination function is ``log'' or ``inv''.
In almost all cases, Naive and Prod methods are unfairer than other methods in terms of the Gini index, and the NSW method is the fairest among all methods.
SW is in the middle between the Naive/Prod methods and the NSW method.
Fig.~\ref{fig:synthetic-dist-log} and \ref{fig:synthetic-dist-inv} show the distributions of user utilities in $10$ experiments for SW and NSW methods for the case of $n = 75$, $m = 50$, and $\lambda = 0.8$.
The utility distributions of NSW seem to be concentrated around the median compared to those of SW.
Therefore, the NSW method is also confirmed to be fair in terms of the Gini index and distributions of the expected resulting matches.

\subsection{Synthetic Data Experiment II: Fairness-Social Welfare Trade-off}\label{sec:experiments:Synthetic_Data_Experiment_II}

\begin{figure}[t]
    \centering
    \includegraphics[width=\linewidth]{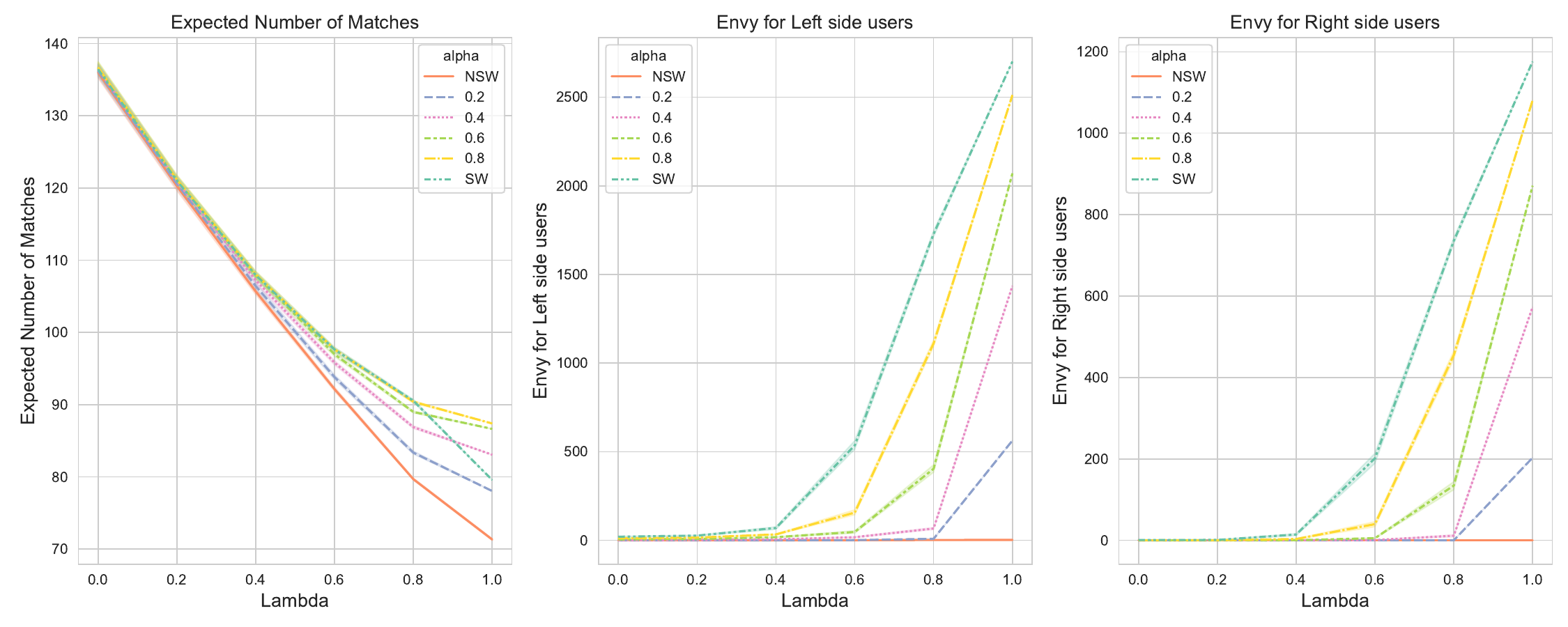}
    \includegraphics[width=0.666\linewidth]{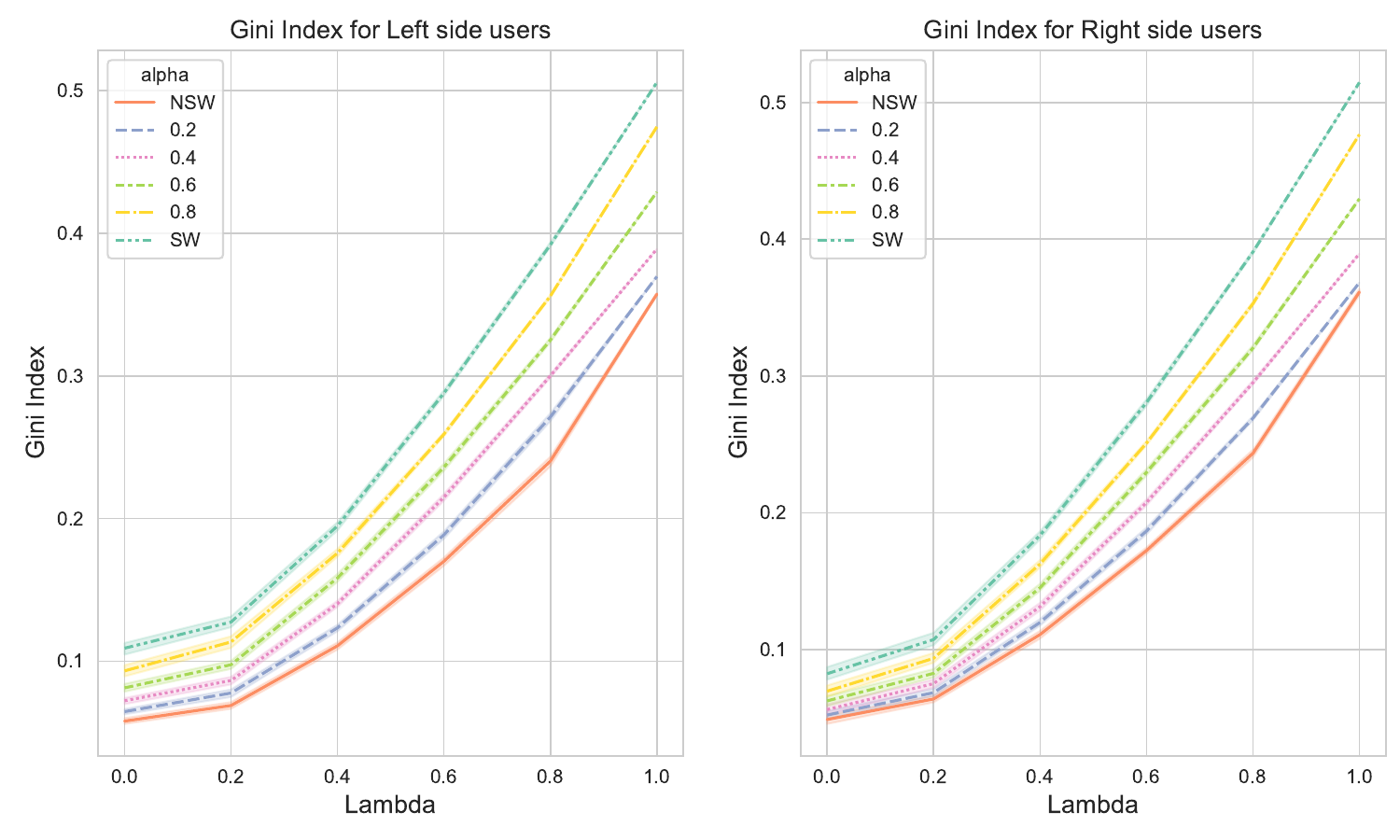}
    \caption{Results of the synthetic data experiments for $\alpha$-SW method. We report the expected number of matches, the number of envy among left/right side users and the Gini Index of left/right side users' utilities. We set $n = 75, m = 50$, $e = \text{``log''}$ and we vary $\lambda \in \{0.0, 0.2, 0.4, 0.6, 0.8, 1.0\}$. Mean value of 10 trials and 95\% CIs are shown, but in most cases CIs are invisible due to their small variations.}
    \label{fig:fairness-efficiency-trade-off}
\end{figure}

In Section~\ref{sec:trade-off-fairness-socialwelfare}, we introduced the $\alpha$-SW maximization method to balance the trade-off between fairness and social welfare.
In this subsection, we report the result of the experiment for the $\alpha$-SW maximization method.

\subsubsection{Experimental setup}
We set the number of left-side users to $n = 75$, the number of right-side users to $m = 50$, and the examination function to $e(k) = 1/\log_2(k+1)$ (``log'').
The preference probabilities $\hat{p}_1(i,j)$ and $\hat{p}_2(j,i)$ for all pairs $(i,j)$ are randomly generated similarly to the first experiment, where $\lambda$ varies in $\{0.0, 0.2, 0.4, 0.6, 0.8, 1.0\}$.
We compute the expected matches, the number of envies, and the Gini index for both sides for the $\alpha$-SW method with $\alpha \in \{0.2,0.4,0.6,0.8\}$, the SW method ($\alpha = 1.0$) and the NSW method ($\alpha \to 0.0$), where we use CVXPY with the CLARABEL solver for the LP problems.
Similarly to the first experiment, we sampled $10$ cases with random seeds of $0,\dots,9$ and reported the average numbers.

\subsubsection{Results}
The results of experiments for the $\alpha$-SW method are reported in Fig.~\ref{fig:fairness-efficiency-trade-off}.
As expected, the expected number of matches is increasing in $\alpha$ with the exception of $\lambda = 1.0$ and $\alpha = 1.0$ (the SW method).
It seems to be that SW is likely to fall into a local optimum when there is a very strong bias of popularity.
Except for such an extreme case, the $\alpha$-SW method with a larger $\alpha$ achieves more expected number of matches.

On the other hand, the $\alpha$-SW with a large $\alpha$ results in unfairness of recommended opportunities in terms of the number of envy and Gini index of user utilities.
Therefore, we can choose $\alpha \in (0,1]$ to balance the trade-off between social welfare and the fairness of the recommended opportunities.

\subsection{Synthetic Data Experiment III: Sinkhorn}\label{sec:experiments:Sinkhorn}

\begin{figure}[t]
    \centering
    \includegraphics[width=\linewidth]{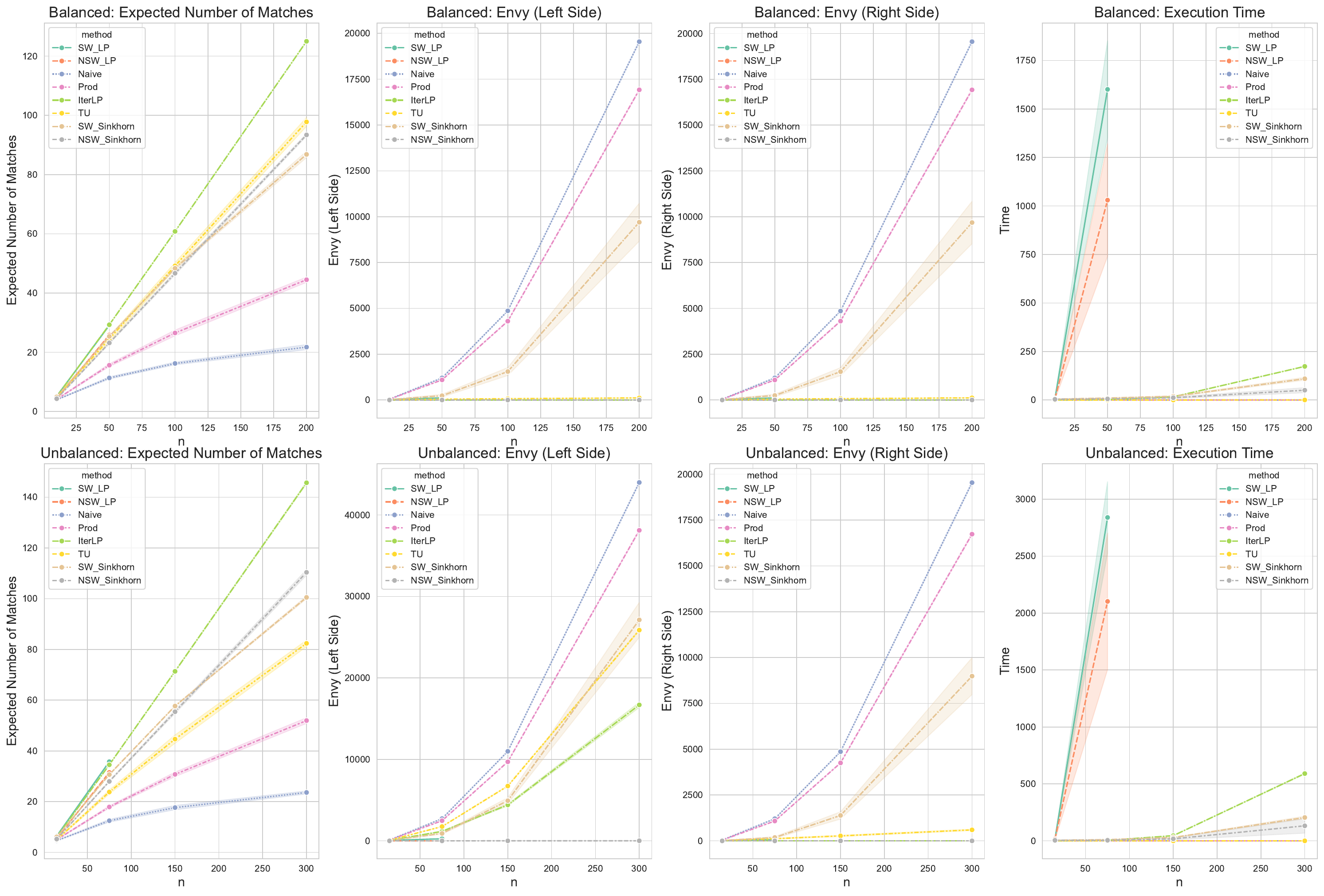}
    \caption{
        Results for the Sinkhorn-based methods (SW\_Sinkhorn and NSW\_Sinkhorn), compared against LP-based and baseline methods. 
        We set $\lambda=0.8$, and $e(k)=1/k$. 
        The first row shows the balanced case $n=m$, while the second row presents the unbalanced case $n=\mathrm{int}\left(\frac{3}{2}\cdot m \right)$.
        The first column reports the expected number of matches, the second and third columns show the number of envious pairs for the right-side and left-side agents, respectively, and the fourth column depicts execution time. Note that the results for SW\_LP and NSW\_LP are limited to instances with $n=10,15,50,75$ due to computational constraints.
    }
    \label{fig:sinkhorn}
\end{figure}

We now evaluate the performance of the Sinkhorn-based approach presented in Section~\ref{sec:Sinkhorn}, comparing it with other methods.

\subsubsection{Experimental setup}
Our experimental evaluation compares eight different methods. Two main approaches are LP-based methods (SW\_LP and NSW\_LP) and Sinkhorn-based methods (SW\_Sinkhorn and NSW\_Sinkhorn), alongside baseline methods (Naive, Prod, IterLP, and TU). 
The LP-based methods implement Algorithm~\ref{alg:alternating_SW_maximization} using an LP solver: SW\_LP employs $F_1(\A,\B)=F_2(\A,\B)=\mathrm{SW}(\A,\B)$, while NSW\_LP employs $F_1(\A,\B)=\log \mathrm{NSW}_1(\A,\B)$ and $F_2(\A,\B)=\log \mathrm{NSW}_2(\A,\B)$.
These LP-based methods were implemented using the CVXPY optimization library with the CLARABEL solver.
In contrast, the Sinkhorn-based methods, SW\_Sinkhorn and NSW\_Sinkhorn, utilize the same objective functions as their LP-based counterparts but employ the Sinkhorn algorithm as described in Algorithm~\ref{alg:alternating_Sinkhorn}.

As in our previous experiments, all methods were implemented in Python. We ran the non-Sinkhorn methods on a CPU with 16GB memory, while the Sinkhorn-based methods were executed on a GPU (NVIDIA T4 with 12GB memory and 15GB VRAM). Notably, implementing the non-Sinkhorn methods on the GPU did not improve performance due to the overhead associated with CPU-GPU data transfers.

We investigated both balanced and unbalanced scenarios in our experiments.
In the balanced scenario, we set $n=m$ with $n \in \{10, 50, 100, 200\}$.
For the unbalanced scenario, we selected $n \in \{15, 75, 150, 300\}$ and defined $m=\mathrm{int}\left(\frac{2}{3}\cdot m \right)$.
Following the methodology outlined in Section~\ref{sec:experiments:Synthetic_Data_Experiment_I}, we generated 10 random instances for each configuration. Here, the popularity parameter was set to $\lambda=0.8$ and the examination function was defined as $e(k)=\frac{1}{k}$. 
For the LP-based methods (SW\_LP and NSW\_LP), experiments were limited to instances with $n\in \{10, 50\}$ (balanced case) or $n \in \{15, 75\}$ (unbalanced case),  along with their corresponding 
$m$ values.
This restriction was imposed since larger instances required computational times exceeding 3000 seconds per run.
For the Sinkhorn-based methods (SW\_Sinkhorn and NSW\_Sinkhorn), the hyper-parameter $\tau$ for Algorithm~\ref{alg:Sinkhorn} is selected from $[100,500]$ by Optuna~\cite{optuna_2019} with TPESampler for each test case.
This tuning process is not included for the execution time reported below.

\subsubsection{Results}
Fig.~\ref{fig:sinkhorn} summarizes the experimental results for the Sinkhorn-based methods, evaluated with respect to three metrics: social welfare (the expected number of matches), fairness (as measured by the number of envious pairs), and execution time.

\paragraph{The expected number of matches}
For smaller instances (up to $n=50$ in the balanced case and $n=75$ in the unbalanced case), IterLP and SW\_LP produced high expected matches. Furthermore, IterLP achieved a high number of expected matches even for larger instances in both balanced and unbalanced cases.
As mentioned in Section~\ref{sec:Sinkhorn}, methods based on the Sinkhorn algorithm are approximate, which may result in lower social welfare (SW) values compared to those obtained through LP-based computations. 
However, in our experiments, the Sinkhorn-based methods demonstrated competitive performance, achieving a high number of expected matches in both the balanced and unbalanced cases.

\paragraph{Fairness}
In the balanced case with $n\le 50$, NSW\_LP exhibited almost zero envy. For $n\le 200$, IterLP, TU, and NSW\_Sinkhorn all maintained similarly low envy in both sides. 
In the unbalanced case, for experiments up to $n=50$, NSW\_LP maintained near-zero envy. 
Notably, NSW\_Sinkhorn consistently achieved near-zero envy even for larger instances up to $n=300$.
While IterLP performed well in terms of the expected number of matches, it exhibited higher envy among the right-side agents (the larger group) in the unbalanced case, with the number of envious pairs exceeding 15,000 on average for $n=300$.
Moreover, the remaining methods exhibited even higher levels of unfairness in such unbalanced scenarios.

\paragraph{Execution time}
Finally, we observe in Fig.~\ref{fig:sinkhorn} that Sinkhorn-based methods are considerably more efficient than LP-based methods.
For both balanced and unbalanced cases, SW\_LP and NSW\_LP required over 1,000 seconds to solve instances with $n=50$ or $n=75$.
In contrast, SW\_Sinkhorn and NSW\_Sinkhorn required only around 200 seconds even for larger instances with $n=200$ or $n = 300$.
Additionally, the computational time for IterLP increased significantly for larger instances; for example, in the unbalanced case with $n=300$, it required more than 500 seconds). These results highlight the scalability advantage of Sinkhorn-based methods, particularly for larger-scale problems.

\subsection{Real-World Data Experiment}\label{sec:experiments:Real-World_Data_Experiment}

To validate the performance of the methods in practical situations, we also report the results of experiments with two real-world dating datasets: one from a Japanese online dating platform and another from a speed dating experiment.

\subsubsection{A Japanese online dating platform data}
One dataset is from a Japanese online dating platform with millions of cumulative members. 
In this service, male users receive recommendation lists of female users and choose to send a ``like'' or not (``dislike'') to each user in the recommended list.
If a male user sends a ``like'' to a female user, she receives a notification and can choose to ``match'' with him or decline (``sorry'').
After a match is made, users can chat and communicate.
The process works similarly for female users, who receive recommendations of male users and respond to them.

To collect a moderate-sized dataset for our experiments and methods, we extracted data consisting of 200 male and 200 female users.
To alleviate the sparsity issue in training our Matrix Factorization-based preference model, we sampled a 200 $\times$ 200 user matrix where users had relatively more interactive actions with each other.
Using ``like''/``dislike'' data of male users and ``match''/``sorry'' data of female users in these sampled data, we compute preference scores $\hat{p}_1(i, j)$ and $\hat{p}_2(j, i)$ using the Alternating Least Squares (ALS) method~\cite{koren2009matrix}.

\subsubsection{Speed dating experiment data}
The other is the \emph{speed dating experiment} dataset, originally collected by \citet{fisman2006gender} from experimental speed dating events conducted from 2002 to 2004, and publicly available on Kaggle~\cite{speeddatingdataset}.
In this experiment, male and female participants engaged in four-minute conversations and then indicated whether they wished to meet their partner again.
The dataset contains $274$ participants of gender 0 and $277$ participants of gender 1.
Each participant engaged in speed dating with approximately 10 participants of the opposite gender within a session and made a decision of whether to ``accept'' each partner.

Based on this dataset, we estimated preference scores using the ALS method and used these scores as the preference inputs for our experiments.
Accordingly, this experiment is conducted using preference scores estimated from real-world behavioral data involving approximately 270 participants on each side.

\subsubsection{Experimental procedure}

For both datasets, we conducted simulation experiments using the preferences explained above.
The examination function is set to either $e(k) = 1/\log(k+1)$ (``log'') or $1/k$ (``inv'').
For each recommendation method, we computed the corresponding recommendation policy, simulated the matching process, and measured the expected number of matches and the number of envies.
The methods compared are Naive, Prod, TU, IterLP, SW\_LP, NSW\_LP, SW\_Sinkhorn, NSW\_Sinkhorn, and $\alpha$-SW\_Sinkhorn ($\alpha = 0.2, 0.4, 0.6, 0.8$), where SW\_LP and NSW\_LP were conducted only for the Japanese online dating dataset because it is infeasible to compute them for the speed dating data given the size of the dataset.
For the Sinkhorn-based methods (SW\_Sinkhorn, NSW\_Sinkhorn, and $\alpha$-SW\_Sinkhorn), the hyperparameter $\tau$ was tuned using Optuna~\cite{optuna_2019} based on the expected number of matches after 5 iterations of Algorithm~\ref{alg:alternating_Sinkhorn} with 10 sampling trials in the range of $[10,10000]$ with log scale, and the value achieving the best performance was adopted.

\subsubsection{Results}\label{sec:experiments:Real_World_Data:results}

\begin{table}[t]
    \centering
    \caption{Results of experiments with the Japanese online dating data}
    \label{tab:realdata}
    \begin{tabular}{llrrrrr}
    \toprule
    e & Method & Expected Matches & \multicolumn{2}{c}{Envy} & \multicolumn{2}{c}{Gini index} \\
     & &  & \multicolumn{1}{c}{Men} & \multicolumn{1}{c}{Women} & \multicolumn{1}{c}{Men} & \multicolumn{1}{c}{Women} \\
    \midrule
    log & Naive & 60.083079 & 1495 & 3016 & 0.574828 & 0.604412 \\
     & Prod & 106.002036 & 765 & 608 & 0.548021 & 0.623122 \\
     & IterLP & 96.826036 & 171 & 75 & 0.515242 & 0.547144 \\
     & TU & 102.694717 & 736 & 695 & 0.559982 & 0.639779 \\
     & SW\_LP & 111.373207 & 434 & 331 & 0.517677 & 0.597365 \\
     & NSW\_LP & 90.388330 & 31 & 14 & 0.380742 & 0.476984 \\
     & SW\_Sinkhorn & 96.593995 & 1606 & 2028 & 0.549705 & 0.639192 \\
     & $\alpha$-SW\_Sinkhorn ($\alpha = 0.8$) & 93.381703 & 672 & 474 & 0.479193 & 0.577669 \\
     & $\alpha$-SW\_Sinkhorn ($\alpha = 0.6$)  & 95.250181 & 333 & 215 & 0.456943 & 0.552828 \\
     & $\alpha$-SW\_Sinkhorn ($\alpha = 0.4$) & 89.091450 & 100 & 70 & 0.427996 & 0.521816 \\
     & $\alpha$-SW\_Sinkhorn ($\alpha = 0.2$) & 86.415991 & 47 & 27 & 0.415157 & 0.510851 \\
     & NSW\_Sinkhorn & 79.041906 & 42 & 19 & 0.396544 & 0.492324 \\
    \cline{1-7}
    inv & Naive & 23.219063 & 2061 & 3561 & 0.714112 & 0.725502 \\
     & Prod & 62.206237 & 942 & 756 & 0.634676 & 0.686189 \\
     & IterLP & 66.544071 & 134 & 51 & 0.555366 & 0.561728 \\
     & TU & 62.470271 & 884 & 866 & 0.644661 & 0.700632 \\
     & SW\_LP & 74.948328 & 330 & 254 & 0.550223 & 0.609980 \\
     & NSW\_LP & 59.373563 & 19 & 8 & 0.382252 & 0.470557\\
     & SW\_Sinkhorn & 51.831140 & 1538 & 2165 & 0.513491 & 0.610453 \\
     & $\alpha$-SW\_Sinkhorn ($\alpha = 0.8$) & 54.697330 & 772 & 861 & 0.466043 & 0.568090 \\
     & $\alpha$-SW\_Sinkhorn ($\alpha = 0.6$) & 56.070438 & 279 & 179 & 0.430077 & 0.512655 \\
     & $\alpha$-SW\_Sinkhorn ($\alpha = 0.4$) & 60.086682 & 140 & 88 & 0.451711 & 0.519969 \\
     & $\alpha$-SW\_Sinkhorn ($\alpha = 0.2$) & 53.275528 & 73 & 47 & 0.432418 & 0.503906 \\
     & NSW\_Sinkhorn & 44.314802 & 45 & 35 & 0.439580 & 0.502645 \\
    \bottomrule
    \end{tabular}
\end{table}
\begin{table}[t]
    \centering
    \caption{Results of experiments with the speed dating experiment data}
    \label{tab:kaggledata}
    \begin{tabular}{llrrrrr}
    \toprule
    e & Method & Expected Matches & \multicolumn{2}{c}{Envy} & \multicolumn{2}{c}{Gini index} \\
     & &  & \multicolumn{1}{c}{Gender 0} & \multicolumn{1}{c}{Gender 1} & \multicolumn{1}{c}{Gender 0} & \multicolumn{1}{c}{Gender 1} \\
    \midrule
    log & Naive & 92.713261 & 14701 & 12494 & 0.490840 & 0.532977 \\
     & Prod & 178.639153 & 5836 & 5623 & 0.545096 & 0.560084 \\
     & IterLP & 180.471494 & 983 & 771 & 0.468215 & 0.479547	 \\
     & TU & 171.696808 & 4768 & 4833 & 0.542474 & 0.555645	 \\
     & SW\_Sinkhorn & 160.038641 & 4417 & 4043 & 0.520178 & 0.533525 \\
     & $\alpha$-SW\_Sinkhorn ($\alpha = 0.8$) & 162.096504 & 3076 & 2803 & 0.498031 & 0.512871 \\
     & $\alpha$-SW\_Sinkhorn ($\alpha = 0.6$)  & 156.624234 & 2042 & 1871 & 0.477870 & 0.491285 \\
     & $\alpha$-SW\_Sinkhorn ($\alpha = 0.4$) & 154.991313 & 734 & 788 & 0.439138 & 0.455036 \\
     & $\alpha$-SW\_Sinkhorn ($\alpha = 0.2$) & 148.641272 & 162 & 214 & 0.410317 & 0.427646 \\
     & NSW\_Sinkhorn & 131.987824 & 17 & 18 & 0.401670 & 0.419427 \\
    \cline{1-7}
    inv & Naive & 9.052270 & 15817 & 13417 & 0.734963 & 0.764334	 \\
     & Prod & 51.346129 & 5840 & 5588 & 0.741575 & 0.746111	 \\
     & IterLP & 83.202273 & 162 & 117 & 0.523850 & 0.527999	 \\
     & TU & 56.838873 & 4013 & 4049 & 0.718418 & 0.720368	 \\
     & SW\_Sinkhorn & 54.836027 & 4205 & 3954 & 0.599449 & 0.604131	 \\
     & $\alpha$-SW\_Sinkhorn ($\alpha = 0.8$) & 61.122392 & 1827 & 1798 & 0.531419 & 0.545869 \\
     & $\alpha$-SW\_Sinkhorn ($\alpha = 0.6$) & 60.184377 & 382 & 514 & 0.451615 & 0.459314	 \\
     & $\alpha$-SW\_Sinkhorn ($\alpha = 0.4$) & 65.095997 & 139 & 170 & 0.454878 & 0.461176	 \\
     & $\alpha$-SW\_Sinkhorn ($\alpha = 0.2$) & 54.634679 & 62 & 103 & 0.488815 & 0.498133	 \\
     & NSW\_Sinkhorn & 51.304954 & 18 & 15 & 0.432176 & 0.447612 \\
    \bottomrule
    \end{tabular}
\end{table}

The results are shown in Table~\ref{tab:realdata} and \ref{tab:kaggledata}.
For the Japanese online dating dataset, SW\_LP achieves the highest expected number of matches in both the ``log'' and ``inv'' cases.
The second-highest method differs depending on the examination function: Prod in the ``log'' case and IterLP in the ``inv'' case.
For NSW\_LP, SW\_Sinkhorn, and $\alpha$-SW\_Sinkhorn (with larger $\alpha$), the expected number of matches is comparable to that of other methods such as Prod, IterLP, and TU.

In contrast, for the speed dating experiment dataset, IterLP attains the highest expected number of matches in both the ``log'' and ``inv'' cases.
In synthetic data experiments, IterLP also achieves a high expected number of matches in many settings, second only to SW\_LP (which is not evaluated in the present experiment), and thus these results are consistent with our prior findings.
On the speed dating dataset as well, SW\_Sinkhorn and $\alpha$-SW (with relatively large $\alpha$) achieve expected numbers of matches comparable to those of existing methods other than IterLP.

With respect to the number of envies, NSW\_LP yields the lowest values on the Japanese online dating dataset, followed by NSW\_Sinkhorn.
On the speed dating experiment dataset as well, NSW\_Sinkhorn achieves the smallest number of envies.
Moreover, these methods also produce more equitable matching outcomes than other approaches in terms of the Gini index.
Therefore, these experimental results suggest that our NSW\_LP and NSW\_Sinkhorn methods are effective for providing fair recommendations even when applied to real-world datasets.

\section{Limitation}
We acknowledge several limitations.
First, we conducted experiments on synthetic data with varying popularity bias and two real-world datasets. However, publicly available large-scale datasets from matching platforms are generally scarce, and validating our methods on data with millions of users remains an important next step.
Second, although the Sinkhorn-based algorithm significantly improves scalability compared to LP-based methods, larger instances still require substantial computation time. Further algorithmic improvements would be needed for deployment on platforms with millions of users.
Third, our framework assumes that preference probabilities are estimated in advance and fixed during the recommendation process. In practice, preferences and recommendations are updated incrementally as new user interactions are observed. Extending our methods to such online settings is a direction for future work.

\section{Conclusion and Future Work}

In this study, we introduced a model of reciprocal recommendations inspired by online dating platforms and formalized the notions of double envy-freeness and Pareto optimality for users on both sides of the platform.
We defined the platform's social welfare as the expected number of matches produced by a given recommendation policy and described an alternating SW maximization procedure.
Our experiments reveal that while the SW maximization approach effectively increases the number of matches, it can lead to significant unfairness, as measured by violations of double envy-freeness.
To address this issue, we proposed an alternating NSW maximization algorithm implemented via the Frank-Wolfe algorithm.
To improve scalability, we further introduced an accelerated approach based on the Sinkhorn algorithm, which significantly reduces computational time while yielding double envy-free outcomes in nearly all cases.

Through experiments conducted on both synthetic and real-world data, we demonstrate that the NSW method achieves nearly zero envy while maintaining competitive levels of social welfare.
These results suggest that our NSW methods are effective for two-sided matching platforms where fairness among users is crucial, such as online dating platforms.

For future work, developing theoretically rigorous exact or approximation algorithms for NSW maximization would be an interesting direction.
Additionally, adapting alternative fairness frameworks---such as Walrasian equilibrium criteria~\cite{xia2019we} and Lorenz efficiency~\cite{Virginie2021}---to our model with explicit mutual interest requirements could provide complementary approaches to fair reciprocal recommendation.
Finally, online evaluation methodologies such as A/B testing would provide valuable insights into the real-world performance of our methods.

\bibliographystyle{ACM-Reference-Format}
\bibliography{main}

\appendix

\section{Notation}\label{appendix:notation}

For the reader's convenience, Table~\ref{tab:notation} summarizes the main notation used throughout this paper.

\begin{table}[htbp]
\centering
\caption{Summary of notation.}
\label{tab:notation}
\begin{tabular}{lcl}
\toprule
Category & Symbol & Description \\
\midrule
\textit{Agents} & $N = \{a_1, a_2, \ldots, a_n\}$ & Set of $n$ agents on the left side \\
 & $M = \{b_1, b_2, \ldots, b_m\}$ & Set of $m$ agents on the right side \\
\midrule
\textit{Preference} & $\hat{p}_1(i,j)$ & Preference probability of agent $a_i$ for agent $b_j$ \\
 & $\hat{p}_2(j,i)$ & Preference probability of agent $b_j$ for agent $a_i$ \\
 & $p_{ij}$ & Product $\hat{p}_1(i,j) \cdot \hat{p}_2(j,i)$ \\
\midrule
\textit{Recommendations} & $A_i \in \mathbb{R}^{m \times m}_{\geq 0}$ & Doubly stochastic matrix for agent $a_i$ \\
 & $B_j \in \mathbb{R}^{n \times n}_{\geq 0}$ & Doubly stochastic matrix for agent $b_j$ \\
 & $A_i(j,k)$ & Prob. that $b_j$ appears at position $k$ in $a_i$'s list \\
 & $(\bm{A}, \bm{B})$ & Recommendation policy \\
\midrule
\textit{Examination} & $e(k)$ & Examination probability at position $k$ \\
\midrule
\textit{Utilities} & $U_i(\bm{A}, \bm{B})$ & Utility (expected matches) of agent $a_i$ \\
 & $V_j(\bm{A}, \bm{B})$ & Utility (expected matches) of agent $b_j$ \\
\midrule
\textit{Welfare} & $\mathrm{SW}(\bm{A}, \bm{B})$ & Social welfare (total expected matches) \\
 & $\mathrm{NSW}_1(\bm{A}, \bm{B})$ & Nash social welfare for left-side agents \\
 & $\mathrm{NSW}_2(\bm{A}, \bm{B})$ & Nash social welfare for right-side agents \\
 & $\mathrm{W}_1^{\alpha}(\bm{A}, \bm{B})$ & $\alpha$-social welfare for left-side agents \\
 & $\mathrm{W}_2^{\alpha}(\bm{A}, \bm{B})$ & $\alpha$-social welfare for right-side agents \\
\midrule
\textit{Parameters} & $\alpha \in (0, 1]$ & Trade-off parameter in $\alpha$-SW method \\
 & $\tau > 0$ & Regularization parameter in Sinkhorn algorithm \\
\bottomrule
\end{tabular}
\end{table}

\section{Additional Experiments}\label{appendix:additional_experiments}
In this appendix, we report the results of two additional experiments: robustness check to preference estimation errors, sensitivity, and effects of hyperparameter $\tau$ in Sinkhorn-based methods.

\subsection{Robustness check to preference estimation errors}\label{appendix:robustness}

To validate the robustness of our methods to preference estimation errors, we added perturbation terms to preference scores.
Specifically, we first generate true preference probabilities $\hat{p}_1(i, j)$ and $\hat{p}_2(i, j)$ as in the synthetic data experiment I (Section~\ref{sec:experiments:Synthetic_Data_Experiment_I}), and then compute perturbed preference probabilities $\tilde{p}_1(i, j)$ and $\tilde{p}_2(i, j)$:
\begin{equation*}
    \tilde{p}_1(i, j) = \min\left\{\max\left\{\hat{p}_1(i, j) + \varepsilon_1(i,j), 0 \right\}, 1\right\}
    \quad\mathrm{and}\quad
    \tilde{p}_2(i, j) = \min\left\{\max\left\{\hat{p}_2(i, j) + \varepsilon_2(i,j), 0 \right\}, 1\right\},
\end{equation*}
where $\varepsilon_1(i, j)$ and $\varepsilon_2(i, j)$ are drawn independently from a normal distribution with mean $0$ and standard deviation $\sigma$.
For each method, the recommendation policies are based on the perturbed preferences $\tilde{p}_1(i, j),~\tilde{p}_2(i, j)$, and simulation processes for evaluation are conducted with the true preferences $\hat{p}_1(i, j),~\hat{p}_2(i, j)$.
Other settings are similar to that of the synthetic experiment I in Section~\ref{sec:experiments:Synthetic_Data_Experiment_I}.
We test the cases with $n = 75$, $m = 50$, $e(k) = 1/\log(k)$ (``log'') or $1/k$ (``inv''), $\lambda = 0.8$ and $\sigma \in \{0.0, 0.1, 0.2, 0.3, 0.4, 0.5\}$, where the case with $\sigma = 0.0$ is equivalent to the setting of the synthetic experiment I.
We generate $10$ sample sets of preferences with random seed $0,\dots,9$, and report the mean values.

The results are shown in Fig.~\ref{fig:robustness-to-estimation-errors}.
As the standard deviation $\sigma$ of error terms grows, the expected numbers of matches decrease in almost all methods, except for the Prod method of the case with the ``inv'' examination.
Although the performances of most methods decline in the case with estimation errors, the simple Prod method might be relatively unaffected with the errors.
The SW method achieves very high level of expected number of matches in most cases, and the performance of the NSW method also does not decline so quickly, comparing to other methods.
On the other hand, the numbers of envies decrease in many methods as the estimation errors grow.
Regardless of estimation errors, the NSW method yield extremely few envies.
Therefore, it is shown that our NSW method is robust to preference estimation errors in terms of the envy-free fairness.

\begin{figure}[htbp]
    \centering
    \includegraphics[width=1.0\linewidth]{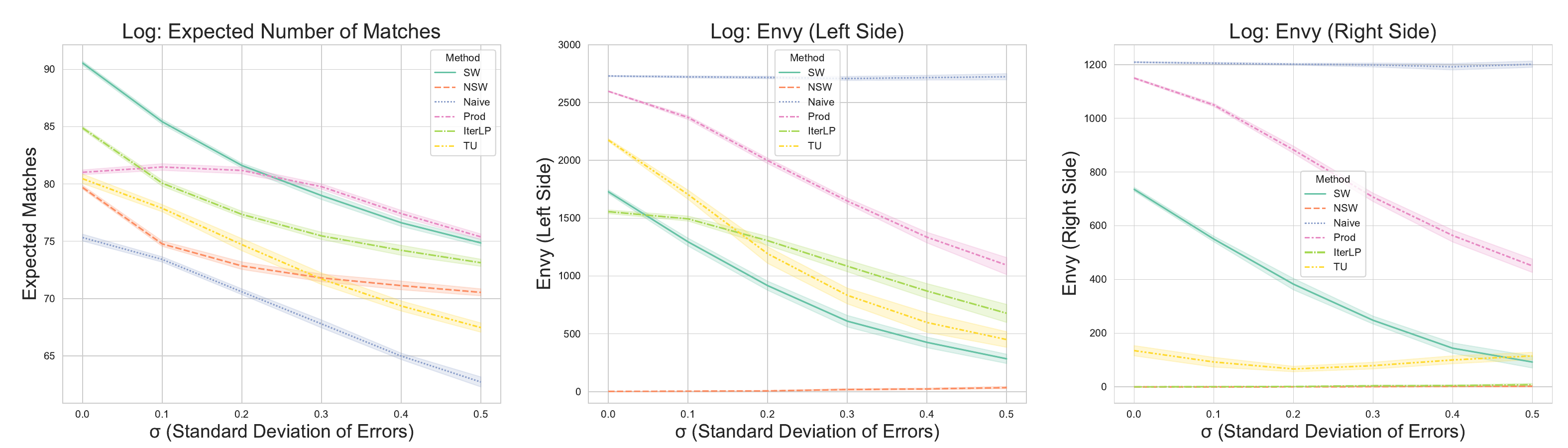}
    \includegraphics[width=0.666\linewidth]{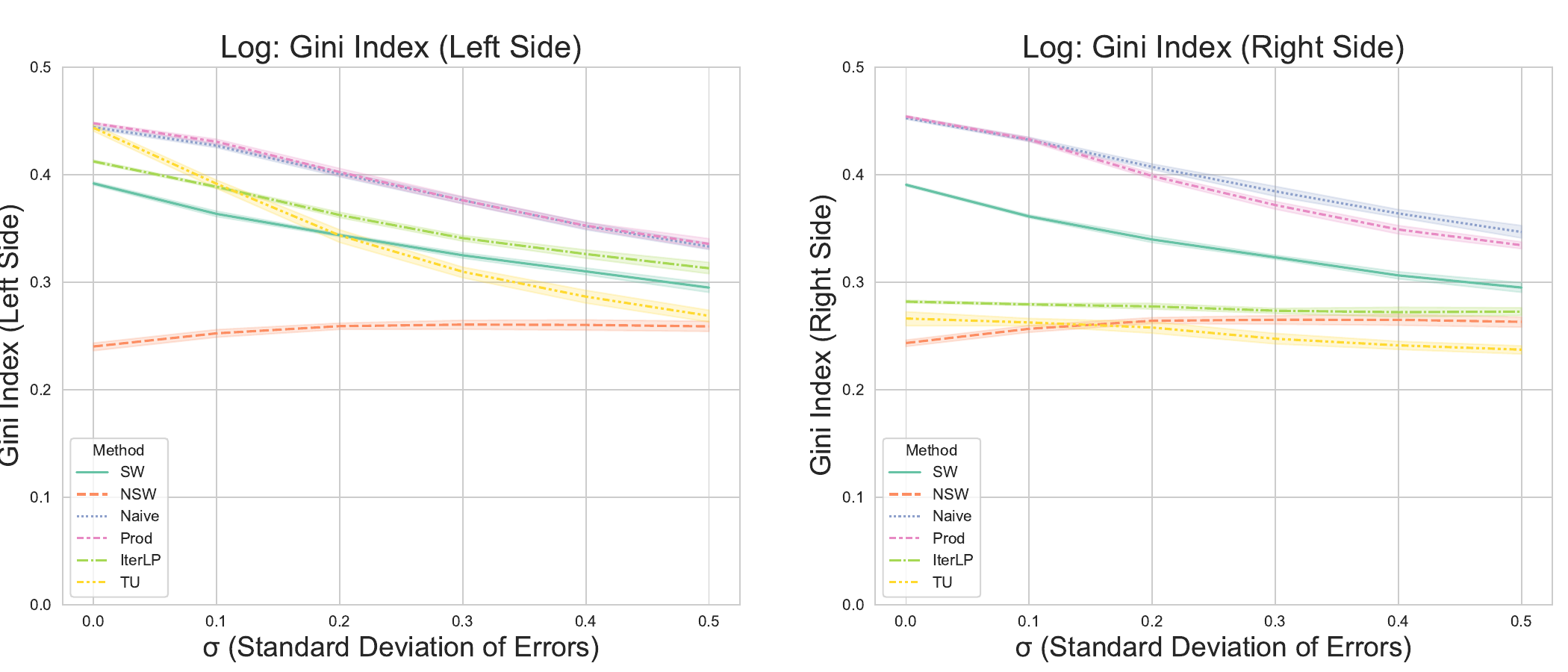}
    \includegraphics[width=1.0\linewidth]{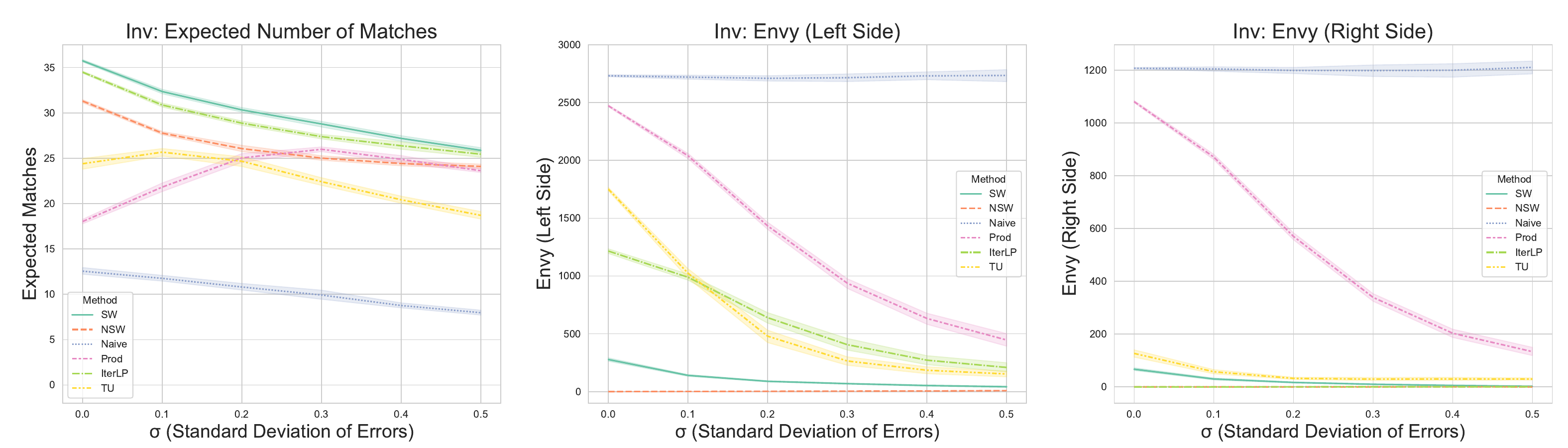}
    \includegraphics[width=0.666\linewidth]{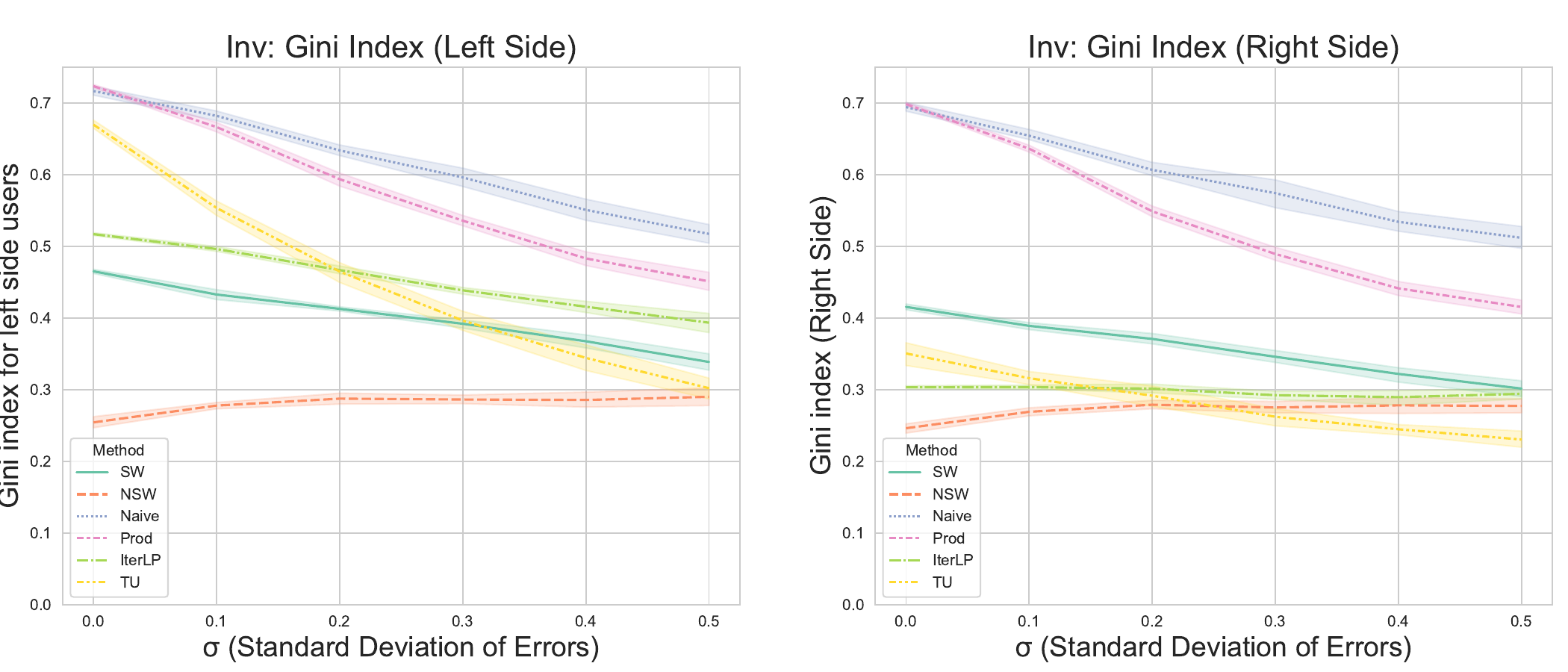}
    \caption{Results for the robustness check experiments to preference estimation errors. We set $n = 75, m = 50, \lambda = 0.8$. The examination functions are $e(k) = 1/\log(k+1)$ (``log'') in the first and second rows, and $e(k) = 1/k$ (``inv'') in the third and fourth rows. We vary the standard deviation of errors $\sigma \in \{0.0, 0.1, 0.2, 0.3, 0.4, 0.5\}$ in the $x$-axis of each graph. We generated 10 samples for each case and reported the average values of the trials and 95 \% confidence intervals.}
    \label{fig:robustness-to-estimation-errors}
\end{figure}

\subsection{Sensitivity to Hyperparameter $\tau$ in Sinkhorn-based algorithms}\label{appendix:sensitivity}

\begin{figure}[htbp]
    \centering
    \includegraphics[width=0.78\linewidth]{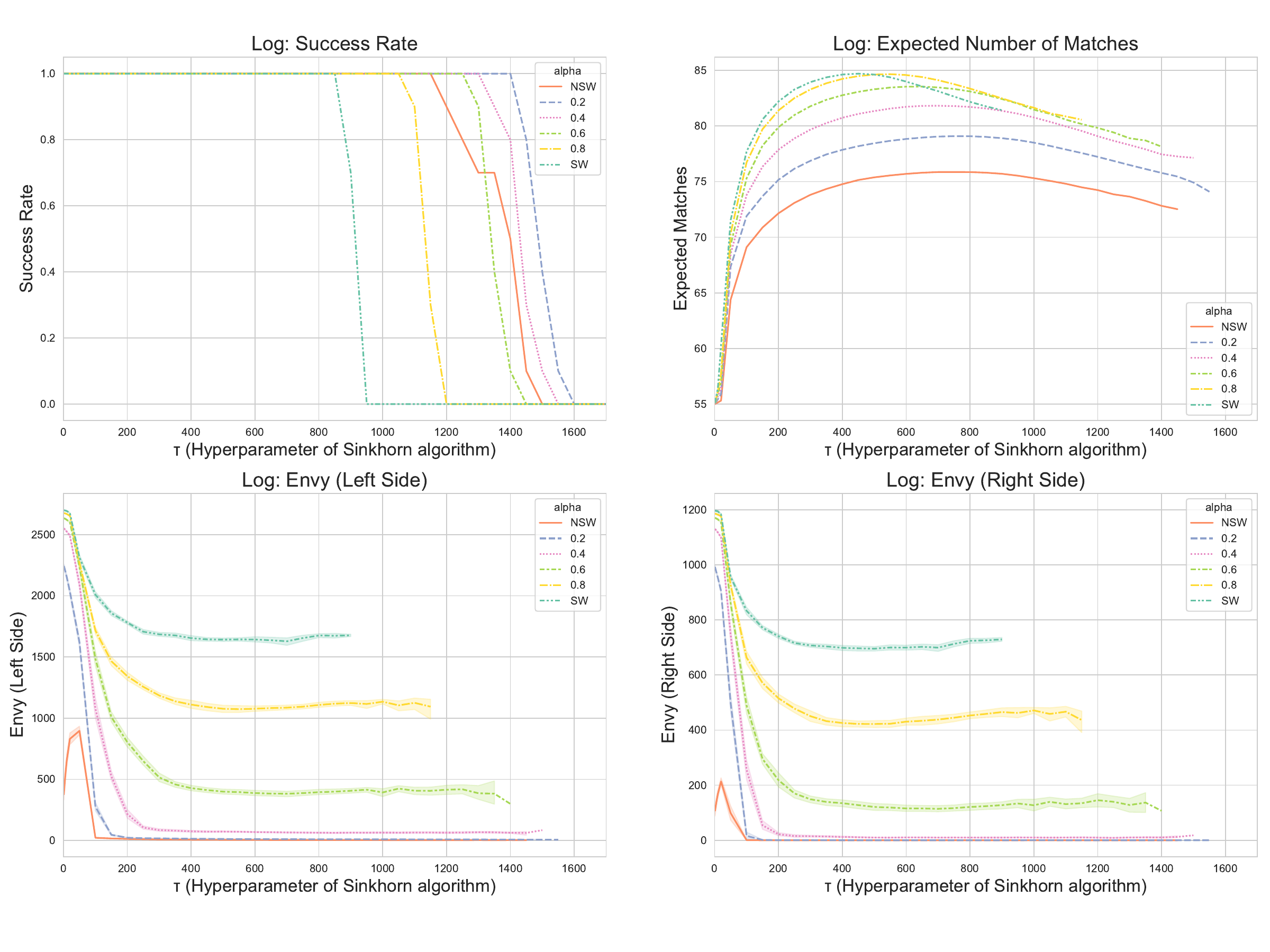}
    \includegraphics[width=0.78\linewidth]{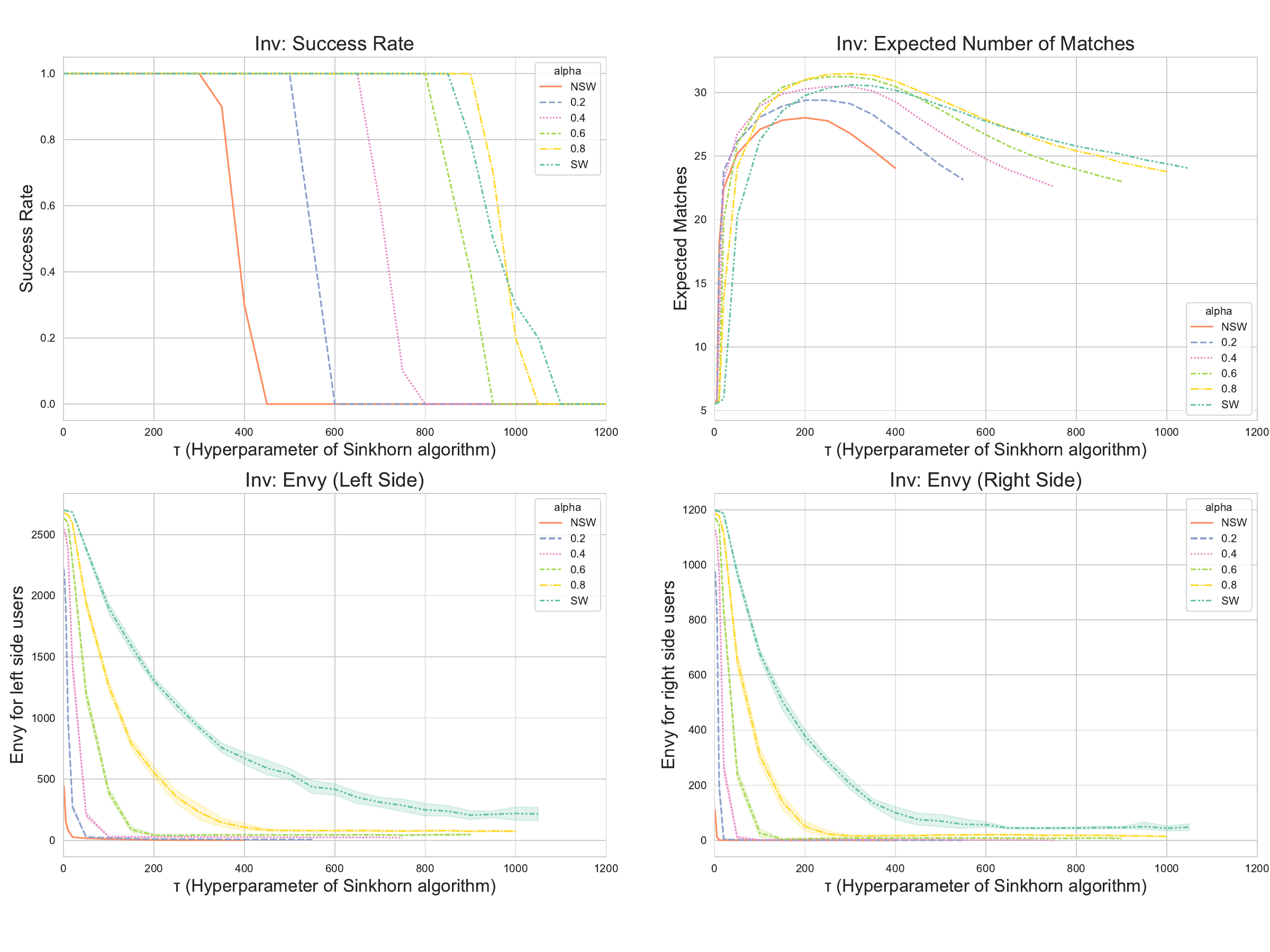}
    \caption{Sensitivity analysis of the Sinkhorn hyperparameter $\tau$. We set $n = 75, m = 50, \lambda = 0.8$. The examination functions are $e(k) = 1/\log(k+1)$ (``log'') in the first and second rows, and $e(k) = 1/k$ (``inv'') in the third and forth rows. We test the SW, NSW and $\alpha$-SW ($\alpha=0.2,0.4,0.6,0.8$) accelerated by the Sinkhorn algorithm. We vary the hyperparameter $\tau \in \{1, 2, 5, 10, 20, 50, 100, \dots,1950,2000\}$ of the Sinkhorn algorithm, which is the $x$-axis of each graph. We generated 10 samples, reported success rates of the 10 trials, and showed the average values and 95 \% confidence intervals of the metrics.}
    \label{fig:sinkhorn-tau-experiment}
\end{figure}

Finally, we report the results of experiments in which the hyperparameter $\tau$ is varied in order to examine its effect in methods based on the Sinkhorn algorithm (Algorithm~\ref{alg:alternating_Sinkhorn}).
The basic experimental setup is the same as in the synthetic data experiment I (Section~\ref{sec:experiments:Synthetic_Data_Experiment_I}), with $n = 75, m = 50, \lambda = 0.8, e(k) = 1/\log(k+1)$~(``log'') or $1/k$ (``inv'').
We evaluate the $\alpha$-SW methods accelerated by the Sinkhorn algorithm, with $\alpha = 0.0$ (NSW), $0.2, 0.4, 0.6, 0.8$ and $1.0$ (SW).
The hyperparameter $\tau$ is set to $1,2,5,10,20,50,\dots,2000$, where $\tau$ is increased in increments of $50$ from $50$ to $2000$.
Similarly to other synthetic data experiments, we sample $10$ sets of preference scores with random seed of $0, \dots, 9$, and report success rates of computation in the 10 trials and average values of expected matches and numbers of envies in succeeded cases.

The results are shown in Fig.~\ref{fig:sinkhorn-tau-experiment}.
First, in our implementation, computation fails due to overflow when $\tau$ takes sufficiently large values. This threshold depends on $\alpha$ and the functional form of $e(k)$.
In addition, the expected number of matches has a unique maximal point as $\tau$ varies, whereas the number of envy decreases almost monotonically as $\tau$ increases.
Therefore, for practical hyperparameter selection, our results suggest that when fairness is prioritized, it is beneficial to choose $\tau$ as large as possible within the computable range, while when the expected number of matches is prioritized, tuning $\tau$ to the value that maximizes this metric is effective.

\end{document}